\documentclass[12pt,draftcls,onecolumn]{IEEEtran}
\usepackage{amsmath,amsfonts,amssymb}
\usepackage{graphicx}
\usepackage{subfigure}
\usepackage{cite}

\newtheorem{theorem}{Theorem}
\newtheorem{definition}{Definition}
\newtheorem{lemma}{Lemma}
\newtheorem{corollary}{Corollary}
\newtheorem{proposition}{Proposition}
\newtheorem{example}{Example}

\newtheorem{implication}{Implication}

\newcommand{\set}[1]{\mathcal{#1}}
\newcommand{\defined}{\triangleq}
\newcommand{\Real}{{\mathbb{R}}}

\newcommand{\field}{{\mathbb F}_q}
\newcommand{\graph}{\set{G}}
\newcommand{\nodes}{\set{P}}
\newcommand{\edges}{\set{E}}
\newcommand{\sessions}{\set{S}}
\newcommand{\sessionsedges}{\set{F}} 

\newcommand{\tail}[1]{\mathrm{tail}(#1)}
\newcommand{\head}[1]{\mathrm{head}(#1)}
\newcommand{\networkCoding}{{\Phi}}
\newcommand{\networkcoding}{{\phi}}
\newcommand{\sourceLocation}{O}
\newcommand{\destinationLocation}{\set{D}} 
\newcommand{\multicastRequirement}{M}
\newcommand{\inputRate}{{\lambda}}
\newcommand{\inputRV}{{U}}
\newcommand{\edgeRate}{{\omega}}
\newcommand{\edgeRV}{{U}}
\newcommand{\normal}{r}
\newcommand{\eht}[2]{#1 \to #2}

\newcommand{\supp}{\Omega}
\DeclareMathOperator{\con}{\overline{con}}
\DeclareMathOperator{\multicastProblem}{{\mathsf{T}}}
\DeclareMathOperator{\indep}{\perp}

\DeclareMathOperator{\kernel}{ker}

\newcommand{\nle}[1]{\stackrel{#1}{\le}}
\newcommand{\nge}[1]{\stackrel{#1}{\ge}}
\newcommand{\nequal}[1]{\stackrel{#1}{=}}

\newcommand{\groundset}{\set{L}}
\newcommand{\espace}{{\cal H}}
\def\myalpha{{{\cal A}}}
\def\mybeta{{\cal B}}
\def\sRV{{S}}
\def\edgeAlphabet{{\cal U}}   
\def\inputRate{{\lambda}}

\newcommand{\N}{\set{N}} 

\newcommand{\where}{:} 
\newcommand{\codeclass}{\set{C}}
\newcommand{\abelian}{\mathrm{ab}}

\title{Dualities Between Entropy Functions and Network Codes}
\author{
\authorblockN{Terence Chan$^{1}$ and Alex Grant}\\
\authorblockA{Institute for Telecommunications Research\\
University of South Australia, Australia \\
{\tt \{terence.chan, alex.grant\}@unisa.edu.au }}
}

\begin{document}
\maketitle


\begin{abstract}
  Characterization of the set of entropy functions $\Gamma^*$ is an
  important open problem in information theory. The region $\Gamma^*$
  is central to the theory of information inequalities, and as such
  could be regarded as a key to the basic laws of information theory.
  Characterization of $\Gamma^*$ has several important
  consequences. In probability theory, it would provide a solution for
  the implication problem of conditional independence. In
  communications networks, the capacity region of multi-source network
  coding is given in terms of $\Gamma^*$. More broadly, determination
  of $\Gamma^*$ would have an impact on converse theorems for
  multi-terminal problems in information theory.  This paper provides
  several new dualities between entropy functions and network
  codes. Given a function $g\geq 0$ defined on all proper subsets of
  $N$ random variables, we provide a construction for a network
  multicast problem which is ''solvable'' if and only if $g$ is the
  entropy function of a set of quasi-uniform random variables. The
  underlying network topology is fixed and the multicast problem
  depends on $g$ only through link capacities and source rates. A
  corresponding duality is developed for linear networks codes, where
  the constructed multicast problem is linearly solvable if and only
  if $g$ is linear group characterizable. Relaxing the requirement
  that the domain of $g$ be subsets of random variables, we obtain a
  similar duality between polymatroids and the linear programming
  bound.  These duality results provide an alternative proof of the
  insufficiency of linear (and abelian) network codes, and demonstrate
  the utility of non-Shannon inequalities to tighten outer bounds on
  network coding capacity regions.
 \end{abstract}
\footnotetext[1]{Terence Chan is also with the Department of Computer
  Science, University of Regina.}

\newpage

\section{Introduction}

Information inequalities are one of the central tools of information
theory. An information inequality is a relation between information
measures such as entropy and mutual information that holds regardless
of the specific choice of joint probability distribution on the
underlying random variables, see \cite[Chapters
12--14]{Yeung02first}. Converse proofs involving chains of information
inequalities are ubiquitous in the literature, extending back to
Shannon. It is somewhat frustrating therefore, that a characterization
of the complete set of information inequalities is lacking. Until the
appearance of the Zhang-Yeung
inequality~\cite{Zhang.Yeung98characterization}, the only known
inequalities were the so-called Shannon, or basic inequalities, being
consequences of the non-negativity of conditional mutual information
(which is a special case of non-negativity of information divergence). Starting with
\cite{Zhang.Yeung97nonshannon}, large classes of conditional
non-Shannon inequalities (e.g. contingent on imposition of certain
Markov constraints) have been
found~\cite{Yeung.Zhang01class,Yeung.Zhang01classisit,Sason02identification,Matus06piecewise}.
A countably infinite class of unconstrained inequalities was reported
in \cite{Makarychev02new}, indexed by the number of random variables
$N$ involved (one inequality for each $N$).  More recently, additional
unconstrained non-Shannon inequalities have been
found~\cite{Dougherty.Freiling.ea06six}. Another countably infinite
class of unconditional inequalities was recently found in
\cite{Matus07infinitely}. This class differs from
\cite{Makarychev02new}, in that a countably infinite number of
inequalities were found for any fixed number of $N\geq 4$ random
variables. As we shall see later, this result has profound
implications.

An intimately related concept is the set of entropy functions
$\Gamma^*$. Let $\espace[\groundset]$ be a subset of a $2^N$
dimensional euclidean space. Each coordinate of this space will be
indexed by a subset of a set $\groundset$ with $N$ elements. Points
$h\in\espace[\groundset]$ can be regarded as functions, mapping from
the set of all subsets of $\groundset$ onto $\Real$ with
$h(\emptyset)=0$.  Points in $\espace[\groundset]$ belong to
$\Gamma^*$ if they correspond to a consistent choice of joint
entropies for a set $\groundset=\{X_1,X_2,\dots,X_N\}$ of $N$ random
variables.  Members of $\Gamma^*$ are called \emph{entropic}, and
members of the closure of $\Gamma^*$, denoted by $\bar{\Gamma}^*$, are
called \emph{almost entropic}.

Characterization of $\bar{\Gamma}^*$ is equivalent to determination of
the set of all possible information inequalities~\cite[Section
12.3]{Yeung02first}. This characterization is lacking for $N>3$. In
contrast, we do know the set $\Gamma\supset\Gamma^*$ corresponding to
the basic inequalities. This set contains some functions that obey the
basic inequalities, but are not entropy functions and do not
correspond to any joint distribution on $N$ random variables. The
basic inequalities are equivalent to the polymatroid axioms, and hence
$\Gamma$ is simply the set of polymatroids, implying a polyhedral
structure.

Characterization of $\Gamma^*$ is an important open problem. It gives bounds for
source coding problems \cite{Yeung.Zhang99distributed}. As shown
in~\cite{Yeung02first}, it would resolve the implication problem of
conditional independence (determination of all additional conditional
independence relations implied by a given set of conditional
independence relationships). In other fields, information inequalities
are also closely linked to group theory~\cite{Chan.Yeung02relation}
and the theory of Kolmogorov
complexity~\cite{Hammer00inequalities,Romashchenko00combinatorial}.
The focus in this paper is however on the link between entropy
functions and the capacity region of multi-source network coding.

The prevailing approach to data transport in communications networks
is based on routing, in which intermediate nodes duplicate and forward
packets towards their final destination. Although such a
store-and-forward scheme is simple to implement, it does not guarantee
efficient utilization of available transmission capacity.  The network
coding approach introduced
in~\cite{Ahlswede.Cai.ea00network,Li.Yeung.ea03linear} generalizes
routing by allowing intermediate nodes to forward packets that are
coded combinations of all received data packets. This seemingly simple
change in approach yields many benefits. Not only can network coding
increase throughput in multicast scenarios, it can also provide
robustness to link failure \cite{Dana.Gowaikar.ea05capacity}, wiretap
security~\cite{Cai.Yeung02secure}, and minimal transmission
cost~\cite{Lun.Ratanakar.ea05minimum-cost}. Naturally, these
advantages are obtained at the expense of increased node complexity.

One fundamental problem in network coding is to understand the
capacity region and the classes of codes that achieve capacity. In the
single session multicast scenario, the problem is well understood. In
particular, the capacity region is characterized by max-flow/min-cut
bounds and linear network codes are sufficient to achieve maximal
throughput
\cite{Li.Yeung.ea03linear,Dougherty.Freiling.ea05insufficiency}.

Significant practical and theoretical complications arise in more
general multicast scenarios, involving more than one session.  It was
recently proved that linear network codes are not sufficient for the
multi-source problem~\cite{Dougherty.Freiling.ea05insufficiency}.
Furthermore, the network coding capacity region is unknown. In fact,
there are only a few tools in the literature for study the capacity
region.

One powerful theoretical tool bounds the capacity region by the
intersection of a set of hyperplanes (specified by the network
topology and connection requirement) and the set of entropy functions
$\Gamma^*$ (inner bound), or its closure $\bar{\Gamma}^*$ (outer
bound)~\cite{Song.Yeung.ea03zero-error,Yeung02first,Yeung.Li.ea06network}.
Recently, these bounds have been tightened to obtain an exact
expression for the capacity region, again in terms of
$\Gamma^*$~\cite{Yan.Yeung.Zhang07capacity}.  Unfortunately, the
capacity region, or even the bounds cannot be computed in practice,
due to the lack of an explicit characterization of the set of entropy
functions for more than three random variables. One way to resolve
this difficulty is via relaxation of the bound, replacing the set of
entropy functions with the set of polymatroids $\Gamma$. The resulting
``linear programming'' bound can be quite loose. Recent
work~\cite{Dougherty.Freiling.ea07matroids} based on matroid theory
showed that application of the Zhang-Yeung
inequality~\cite{Zhang.Yeung98characterization} yields a
tighter bound for the capacity region (by obtaining a better outer
bound for the set of entropy functions).

The main results of this paper are new dualities between non-negative
functions $g\in\espace[\groundset]$ and network codes. These duality
results are based on the construction of a special network multicast
problem from functions $g$.  The underlying network topology is fixed
and the multicast problem depends on $g$ only through the assignment
of link capacities and source rates.

Three main kinds of duality are considered, corresponding to different
restrictions on $g$ and different kinds of network codes.  First, we
show in Theorem~\ref{thm:firstDuality} that the constructed multicast
problem is solvable (i.e. the constructed source rates and link
capacities are in the capacity region) if and only if $g$ is the
entropy function of a set of quasi-uniform random variables. This
duality is extended in Theorem~ \ref{thm:FirstDualityExtension} to
show that the multicast problem is asymptotically solvable with
$\epsilon$ error if and only if $h$ is almost entropic.

The second duality restricts attention to linear network codes. We
show that the multicast problem is linearly solvable if and only if
$g$ is linear group characterizable (i.e.  $g$ is an entropy function
for random variables generated by vector spaces). A corresponding
limiting form of this duality is also provided.

Finally, by relaxing the requirement that the domain of $g$ be subsets
of random variables, we obtain a duality between polymatroids and the
linear programming bound.

These duality results yield several immediate implications. In
particular, we provide an alternative proof to
\cite{Dougherty.Freiling.ea05insufficiency,Dougherty.Freiling.ea07matroids}
for the insufficiency of linear (and abelian) network codes, and
demonstrate the utility of non-Shannon inequalities to tighten outer
bounds on network coding capacity regions.

The paper is organized in the following way.  Section
\ref{sec:networks} introduces some fundamentals of network
coding. Section \ref{sec:algebraic} focuses on network codes with
algebraic structure, and random variables generated by groups with a
variety of algebraic structures. We establish a relation between
linear network codes and random variables generated by vector spaces
and generalize this idea to define the concept of a group network
code. A central theme of the paper is the trade-off between source
rate and link capacity using network coding, i.e. determination of the
network coding capacity region.  Section \ref{sec:tradeoff} introduces
the definitions for admissibility and achievability in the network
coding context. Section \ref{sec:bounds} introduces the concept of
pseudo-variables, which generalize random variables in such a way that
allows a notational unification of the linear programming bound with
that of~\cite{Song.Yeung.ea03zero-error}.

Section \ref{sec:result} proves the duality results, Theorems
\ref{thm:firstDuality} -- \ref{thm:LPbd}. These results rely on the
construction in Section \ref{sec:construction} of a special network
and multicast problem from a function $g$.
Section~\ref{sec:firstDuality} gives the duality between entropic
functions and solvable multicast problems. Section
\ref{sec:secondDuality} provides the corresponding duality for
linearly solvable multicast problems.  These duality results are
extended in Section \ref{sec:thirdDuality} to give a similar link between
polymatroids and the linear programming bound, i.e. a function $g$ is
a polymatroid if and only if the constructed source rates and link
capacities satisfy the bound. This result relies heavily on the notion
of pseudo-variables introduced in Section \ref{sec:bounds}, and in
particular on extension and adhesion of sets of pseudo-variables,
discussed in Appendix \ref{app:LPbd}. Finally, in Section
\ref{sec:fourthDuality} we give a one-way relation between the LP
bound for linear codes, and polymatroids which also satisfy the
Ingleton inequality.

Section \ref{sec:applicationA} explores the implications of our
results, which include the insufficiency of linear or even (abelian)
group network codes, and the necessity for non-Shannon inequalities
for determination of the network coding capacity region.

\emph{Notation:} For a set $\set{A}$, the power set
$2^\set{A}=\{\set{B}\where\set{B}\subseteq\set{A}\}$ denotes the set
of all subsets of $\set{A}$. Given a set of $|\set{A}|$ variables
$\{X_a, a\in\set{A}\}$, and a subset $\set{C}\subseteq\set{A}$, the
subscript $X_{\set{C}}$ shall mean $\{X_c:c\in\set{C}\}$. In contrast,
the notation $Y_{[\set{B}]}$ will be used to index a single variable
out of a set of $2^{|\set{A}|}$ variables $\{Y_{[\set{B}]} :
\set{B}\in2^{\set{A}}\}$. Other notation will be introduced as
necessary throughout the paper.

\section{Networks, Codes and Capacity}\label{sec:networks}

A directed acyclic graph $\graph = (\nodes, \edges)$ is commonly used
as a simplified model of a communication network. The nodes
$u\in\nodes$ and directed edges $e=(\tail{e},\head{e})\in\edges$
respectively model communication nodes and directed, error-free
point-to-point communication links. The terms graph and network will
be used interchangeably. For edges $e,f\in\edges$, write $\eht{f}{e}$
as shorthand for $\head{f}=\tail{e}$.  Similarly, for an edge
$f\in\edges$ and a node $u\in\nodes$, the notations $\eht{f}{u}$ and
$u\rightarrow f$ respectively denote $\head{f}=u$ and $\tail{f}=u$.
So far we have only specified the basic network topology. The
communication problem is specified via imposition of a connection
requirement.

\begin{definition}[Connection Requirement]
  For any network $\graph$, a \emph{connection requirement}
  $\multicastRequirement=(\sessions, \sourceLocation,
  \destinationLocation)$ is specified by three components representing
  the sessions, originating nodes and destination nodes as follows.
  $\sessions$ is an index set of independent multicast sessions, each
  of which is a collection, or stream of data packets to be multicast
  to a prescribed set of destination nodes.
  $\sourceLocation:\sessions\mapsto\nodes$ is a source-location
  mapping, where $\sourceLocation(s)$ is the originating node for
  multicast session $s$.  $\destinationLocation:\sessions\mapsto
  2^{\nodes}$ is a receiver-location mapping, where
  $\destinationLocation(s)\subseteq\nodes$ is the set of nodes
  requiring the data of session $s$.
\end{definition}

It should be noted that there is \emph{no specified rate
  requirement}. The connection requirement differs from the usual
concept of multicast requirement in that it only specifies
\emph{which} nodes require data from which other nodes, and not any
particular desired information rate.

Given a connection requirement $\multicastRequirement$, the goal of a
network code is to efficiently multicast data for session $s$
originating at node $\sourceLocation(s)$ to all receivers in the set
$\destinationLocation(s)$. Nodes are assumed to have sufficient
computing power to implement any desired network coding scheme.

Let $\sessionsedges=\sessions\cup\edges$.  For a network $\graph$ and
connection requirement $\multicastRequirement$, a network code is
specified by a set of source and edge alphabets
$ \left\{\edgeAlphabet_f, f\in\sessionsedges\right\}$ and
a set of local coding functions
\begin{equation*}
\networkCoding \triangleq \left\{
  \networkcoding_e:\prod_{f\in\sessionsedges:f\rightarrow e}
  \edgeAlphabet_f\mapsto
  \edgeAlphabet_e \where e \in \edges \right\}
\end{equation*}
where for ease of notation, $\eht{s}{e} $ indicates
$\sourceLocation(s) \rightarrow e$, and
$f\in\sessionsedges:f\rightarrow e$ means any source or edge incident
to edge $e$.

Data transmission takes place as follows.  Session $s\in\sessions$
generates a source symbol $\inputRV_s$, which is assumed to be
independent of other sessions and uniformly distributed over
$\edgeAlphabet_s$.  The link symbol transmitted along $e\in\edges$ is
$\edgeRV_e = \networkcoding_e(\edgeRV_f \where f\in\sessionsedges,
\eht{f}{e})$. In other words, the symbol transmitted along an outgoing
link of a node is a function of the available sources and incident
link symbols.

We will refer to a network code by $\networkCoding$, with the set of
alphabets $\left\{\edgeAlphabet_f, f\in\sessionsedges\right\}$
implicitly defined. Since the input and link symbols are random
variables, we can also refer to the code by the set of random
variables $\edgeRV_\sessionsedges$, where their joint distribution is
implied by $\networkCoding$.  Clearly,
\begin{align*}
H(\inputRV_\sessions) &= \sum_{s\in\sessions} H(\inputRV_s)=
\sum_{s\in\sessions} \log |\edgeAlphabet_s| \quad\text{and}\\
H(\edgeRV_e) &\le
\log |\edgeAlphabet_e|.
\end{align*}

For a given network code $\networkCoding$ designed for a network
$\graph$ with connection requirement $\multicastRequirement$, the
error probability $P_e(\networkCoding)$ is defined as the probability
that at least one receiver $d\in\bigcup_{s\in\sessions} \destinationLocation(s)$ fails to
correctly reconstruct one or more of its requested source messages
$\{\inputRV_s \where \destinationLocation(s)=d\}$. A \emph{zero-error} network code is one for
which $P_e(\networkCoding)=0$, implying that the source symbols $\edgeRV_s$
are deterministic functions of the corresponding receiver-incident
edge symbols.

\subsection{Algebraic network codes}\label{sec:algebraic}
The above formulation imposes no restriction on the choice of
alphabets and local coding functions. However, in practice, it may be
preferable to impose algebraic structure to reduce the complexity of
encoding and decoding. The overwhelming majority of codes studied for
the point-to-point channel are in fact linear, and linear codes are
also of particular interest in the network coding context.

\begin{definition}[Linear Network Code]
  A network code $\networkCoding$ is \emph{linear} over a finite field
  $\field$ if all source and link alphabets $\edgeAlphabet_f$ are
  vector spaces over some finite field $\field$, and all the local
  encoding functions $\networkcoding_e$ are linear.
\end{definition}
Clearly, for a linear network code, each source alphabet is a vector
subspace and the symbol transmitted along link $e\in\edges$ is a
linear function of the inputs $\edgeRV_\sessions$.  As will be stated
in Proposition~\ref{prop:linearCodeAndLinearChar}, the set of all the
kernels of these linear functions associated with all the links can be
used to ``construct'' the set of source and link random variables
defining the network code.  To understand this relationship, we first
review the construction of random variables from a finite group and
its groups~\cite{Chan.Yeung02relation}.

\begin{definition}[Construction of random variables from subgroups]
  Suppose that $U$ is a random variable uniformly distributed over a
  group $G$. For any subgroup $G_i$, the set of left cosets of $G_i$
  forms a partition in $G$. Let $\set{U}_i$ be an index set of the
  cosets of $G_i$ in $G$. We can define a random variable $U_i$ as a
  function of $U$ such that $U_i$ is the index of the coset of $G_i$
  that contains $U$, or simply that $U_i$ is the coset of $G_i$ that
  contains $U$. The resulting random variable is said to be
  \emph{constructed} from $G$ and $G_i$.
\end{definition}

\begin{definition}[Group characterizable random variables]
\label{df:groupCharacterizable}
A set of random variables $\{U_1, \dots, U_N\}$ (and its induced
entropy function) is called \emph{group characterizable} if it is
equivalent\footnote{Two sets of random variables $\{U_1,\cdots,U_N\}$
  and $\{V_1,\cdots,V_N\}$ with probability distributions $P_U$ and
  $P_V$ respectively are ``equivalent'' if for each $i=1,\cdots, N$,
  there is a one-to-one mapping $\tau_i$ from the support of $U_i$ to
  the support of $V_i$ such that $P_U(U_1,\cdots, U_N) =
  P_V(\tau_1(U_1),\cdots, \tau_N(U_N))$. In this paper, two sets of
  equivalent random variables will be regarded as identical.}  to a
set of random variables constructed from a finite group $G$ and its
subgroups $G_1 ,\cdots, G_N$.

If $G$ is abelian, then $\{U_1 , \cdots, U_N\}$ (and the entropy
function) is called \emph{abelian group characterizable}. If in
addition $G$ and $G_1 ,\cdots, G_N$ are all vector spaces, then the
set of random variables (and the entropy function) is called
\emph{linear group characterizable}.
\end{definition}

Denote the set of group characterizable entropy functions by
$\Gamma^*_{G}\subset\Gamma^*$, the set of abelian group
characterizable functions by $\Gamma^*_\abelian$ and the set of linear
(with respect to a finite field $\field$) group characterizable
functions by $\Gamma^*_{L(q)}$. Then, it is clear that
$\Gamma^*_{L(q)}\subset\Gamma^*_\abelian\subset\Gamma^*_{G}\subset\Gamma^*$.

Random variables constructed from subgroups have been shown to have
many interesting properties. For example, suppose $\{U_1 , \cdots,
U_N\}$ is constructed from a finite group $G$ and its subgroups $G_1
,\cdots, G_N$. Then $ H\left(U_\alpha\right) = \log
|G|/|\bigcap_{i\in\alpha} G_i|$ for any non-empty subset
$\alpha\subseteq\N\defined\{1, 2,\dots,
N\}$\cite{Chan.Yeung02relation}. It was also proved in
\cite{Chan.Yeung02relation} that a linear information inequality is
valid if and only it is satisfied by all group characterizable random
variables. Thus group characterizable random variables have an
interesting role to play in the proof of information inequalities,

Before describing some additional properties of group characterizable
random variables, we will need the concept of quasi-uniform random
variables.
\begin{definition}[Quasi-uniform random variable]
  A discrete finite random variable $U$ defined on a sample space
  $\cal U$ is called \emph{quasi-uniform} if and only if it is
  uniformly distributed over its support $\supp(U)$. In other words,
  the probability distribution of $U$ has the following form:
  \begin{equation*}
    \Pr(U=u) =
    \begin{cases}
      1/| \supp(U)| & \quad \mbox{if } u\in\supp(U) \\
      0 & \quad \mbox{otherwise }
    \end{cases}
  \end{equation*}
  Hence, $H(U) = \log | \supp(U)|$.

  Similarly, a set of random variables $U_1, U_2, \dots, U_N$ (and its
  induced entropy function) is called quasi-uniform if and only if
  every subset of random variables $U_\alpha,
  \alpha\subseteq\{1,2,\dots,N\}$ is quasi-uniform, i.e. $H(U_\alpha)
  = \log |\supp(U_\alpha)|$.
\end{definition}

\begin{lemma}[\cite{Chan.Yeung02relation}, \cite{Chan01combinatorial}]
  Random variables induced by groups and subgroups are
  quasi-uniform. Hence
  $$\Gamma^*_{L(q)}\subset\Gamma^*_\abelian\subset\Gamma^*_{G}\subset\Gamma^*_{Q}\subset\Gamma^*$$
  where $\Gamma^*_{Q}$ is the set of all quasi-uniform entropy
  functions.
\end{lemma}

\begin{figure}[htbp]
  \begin{center}
    \includegraphics*[scale=0.8]{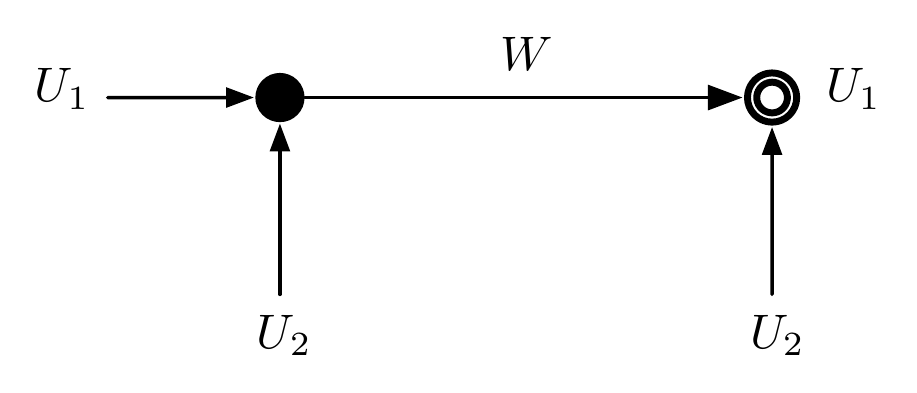}
  \end{center}
  \caption{The side-information network.}\label{fig:sideinfo}
\end{figure}
\begin{lemma}\label{lemm:qucoding}
  With reference to Figure \ref{fig:sideinfo}, consider a simple
  coding problem in which there is a transmitter (indicated by an open
  circle) and a receiver (indicated by a double circle) connected by a
  noiseless point-to-point link. A source $U_1$ is available at
  the transmitter, while correlated side-information $U_2$ is
  available at both transmitter and receiver. The coding problem is to
  encode $U_1,U_2$ into a symbol $W$ defined on the sample space
  $\set{W}$ such that $U_1$ can be constructed perfectly at receiver
  from $W$ and $U_2$.

  Suppose that $\{U_1, U_2\} $ is quasi-uniform. Then one can have a
  zero-error code with rate $\log | \supp(U_1,U_2)|/ | \supp(U_2)| =
  H(U_1|U_2)$, where the code rate is defined as $\log |\set{W}|$.
\end{lemma}
\begin{proof}
  Since $U_2$ is available to both transmitter and receiver, $U_1$ can
  be reconstructed perfectly if the transmitter only sends the index
  of $u_1$ in the set $\{u_1 : (u_1, u_2) \in \supp(U_1,U_2) \}$ for
  any given $u_2 \in \supp(U_2)$.  By the quasi-uniformity of $\{U_1,
  U_2\} $, the cardinality of the set $\{u_1 : (u_1, u_2) \in
  \supp(U_1,U_2) \}$ is $|\supp(U_1,U_2)|/|\supp(U_2)|$ for any $u_2
  \in \supp(U_2)$. Hence, one can easily construct a zero-error code
  at a rate of $\log | \supp(U_1,U_2)|/ | \supp(U_2)| = H(U_1|U_2)$
  that solves the coding problem.
\end{proof}

If the group and subgroups in question possess additional algebraic
properties, the induced random variables may also satisfy certain
additional properties. One interesting example, proved in
\cite{Chan1998,Chan07group} is given as follows.
\begin{proposition}[Ingleton's inequality]\label{prop:Ingleton}
  Suppose that the set of random variables $\{U_1,\dots, U_N\}$ is
  abelian group characterizable. Let $\{V_1,V_2,V_3,V_4\}\subseteq
  \{U_1,\dots, U_N\}$. Then
\begin{equation}\label{eq:ingleton}
  g (1,2)+g(1,3) + g(1,4)
  + g(2,3) + g(2,4) \geq  g(1) + g (2) +g(3,4)
  + g (1,2,3)+g(1,2,4)
\end{equation}
where $g(\alpha)\defined H(V_\alpha)$.
\end{proposition}

\begin{proposition}\label{prop:linearCodeAndLinearChar}
  Suppose that a set of random variables $\{\edgeRV_f,
  f\in\sessionsedges\}$ defines a zero-error linear network
  code.  Then $\{\edgeRV_f, f\in\sessionsedges\}$ is
  linear group characterizable.
\end{proposition}
\begin{proof}[Proof Sketch]
  Suppose that $\networkCoding=\{\networkcoding_e, e\in\edges\}$ is a
  zero-error linear network code with inputs
  $\edgeRV_s\in\edgeAlphabet_s$ for $s\in\sessions$ and link symbols
  $\edgeRV_e\in\edgeAlphabet_e$ for $e\in\edges$.  We will now
  construct a linear group characterization for the set of source/link
  random variables induced by $\networkCoding$. Let
  \begin{enumerate}
  \item $G$ be the vector space formed by the Cartesian product of
    $\prod_{s\in\sessions}\edgeAlphabet_s$;
  \item $\psi_s: G \mapsto\edgeAlphabet_s$ be a linear
    function such that $\psi_s( \edgeRV_s:s\in\sessions)=
    \edgeRV_s$;
  \item $\psi_e: G \mapsto\edgeAlphabet_e$ be a linear function such
    that $\edgeRV_e = \psi_e(\edgeRV_s:s\in\sessions)$; (\emph{This is
      possible as all local coding functions $\networkcoding_e$ are
      linear})
  \item $G_f$ is the kernel of $\psi_f$, denoted by $\kernel(\psi_f)$, for
    $f\in\sessions\cup\edges$. Hence, $G_f$ is a subspace of $G$.
  \end{enumerate}
  Then it is straightforward to show that for any
  $(\edgeRV_s:s\in\sessions)$ and $f\in\sessionsedges$, the value of
  $\psi_f(\edgeRV_s:s\in\sessions)$ can be uniquely
  determined from the index of the coset of $G_f$ that contains
  $(\edgeRV_s:s\in\sessions)$ and vice versa. In other words, the link
  random variable $\edgeRV_f$ is equivalent to the one induced by the
  subspace $G_f$.
\end{proof}

A natural interpretation of Proposition
\ref{prop:linearCodeAndLinearChar} is that linear network codes are
those codes whose induced source and link random variables can be
characterized by a vector space and its subspaces. Developing this
line of thought more generally, we make the following definition.
\begin{definition}[Group network code]
  A \emph{group network code} is a network code $\{\edgeRV_f,
  f\in\sessionsedges\}$ whose source and link random variables are
  induced by a finite group $G$ with subgroups $G_f,
  f\in\sessionsedges$. Furthermore, a group network code is called
  abelian if $G$ is abelian.
\end{definition}
For a group network code $\networkCoding = \{\edgeRV_f,
f\in\sessionsedges\}$, encoding at intermediate nodes works as
follows. Suppose that the source and link random variables
$\{\edgeRV_f, f\in\sessionsedges\}$ are characterized by a finite
group and its subgroups $G_f$ for $ f\in\sessionsedges$.  For
any $f\in\sessionsedges$, let $\edgeAlphabet_f$ be the index set
for the set of left cosets of $G_f$ in $G$. Each edge $e$ receives
symbols $\{\edgeRV_f \where \eht{f}{e}\}$, which are indexes of cosets $G_f$
in $G$. The symbol $\edgeRV_e$ to be transmitted along edge $e$ is the
index of the left coset $G_e$ that contains the intersection of the
cosets of $G_f$ indexed by $\{\edgeRV_f \where \eht{f}{e}\}$.

In fact, in the special case when the group and all its subgroups are
vector spaces, we can index the coset of $G_e$ as elements in a vector
space such that $\edgeRV_e$ is indeed a linear function of
$\{\edgeRV_f \where \eht{f}{e}\}$.

\begin{example}
  An $R$-module generalizes the concept of vector space,
  where the scalars are a members of a ring $R$, instead of a field. It
  consists of an abelian group $K$, and an operation of left
  multiplication by each element in $R$. In particular, for all $r,s \in
  R$ and $g,h \in K$,
  \begin{align*}
    rg &\in K \\
    (rs)g &= r(sg) \\
    (r+s)g &= rg + sg \\
    r(g+h) &= rg + rh \\
    0g &= 0.
  \end{align*}
  $R-module$ codes have been proposed as generalizations of linear
  network codes \cite{Dougherty.Freiling.ea05insufficiency}.  Messages
  to be transmitted along edges are elements in $K$.  The only
  difference is that local encoding functions must be of the form
 \begin{equation*}
   \edgeRV_e = \sum_{f\in\sessionsedges: f \to e} r_{fe} \edgeRV_f
 \end{equation*}
 where $r_{fe} \in R$. As such, there exists elements $M_{es}\in R$
 such that
 \begin{equation*}
   \edgeRV_e = \sum_{s\in \sessions} M_{es} \edgeRV_s.
 \end{equation*}
 Let $G$ be the $|\sessions|$-fold Cartesian product of $K$. For all
 $e\in\edges$ and $s\in\sessions$, let
 \begin{align*}
   G_e &= \left\{(\edgeRV_s \in K :s\in\sessions) : \sum_{s\in
       \sessions} M_{es} \edgeRV_s = 0 \right\}\\
   G_s &=\left\{(\edgeRV_s\in K :s\in\sessions)
   : \edgeRV_s = 0 \right\}.
 \end{align*}

 Then it is straightforward to show that $G_f$ is an abelian subgroup
 of $G$ for $f\in\sessionsedges$ and that the source and link random
 variables induced by the $R-module$ code is characterized by the
 subgroup $G$ and its subgroups $G_f$, $f\in\sessionsedges$.
\end{example}

\subsection{The source rate-link capacity tradeoff}
\label{sec:tradeoff}
So far, we have only considered networks, and codes designed to meet
particular connection requirements. Typically however, each link has
limited capacity, and a fundamental design consideration is the
tradeoff between supportable network throughput and link
capacities. Of primary interest is determination of the minimal link
capacities $\edgeRate\defined(\edgeRate_e: e \in \edges)$ required to
transmit sources over a network at given rates
$\inputRate\defined(\inputRate_s : s \in \sessions)$ such that all
receivers can reconstruct their desired messages with no, or
arbitrarily small probability of error.

\begin{definition}[Admissible rate-capacity tuple]\label{df:admissible}
  Given a network $\graph=(\nodes,\edges)$ and a connection
  requirement $\multicastRequirement$, a rate-capacity tuple
  $(\inputRate , \edgeRate)$ is \emph{admissible} if there exists a
  zero-error network code $\networkCoding = \{\edgeRV_f,
  f\in\sessions\cup \edges\}$, such that
  \begin{align*}
   H(\edgeRV_e) \le \log |\edgeAlphabet_e| &\le \edgeRate_e,
   \quad\forall e \in \edges,\\
   H(\inputRV_s)= \log |\edgeAlphabet_s| &\ge \inputRate_s,
   \quad\forall s\in \sessions,
  \end{align*}
  where $\edgeRV_e$ is the message symbol transmitted along link $e$
  and $\inputRV_s$ is the input symbol generated at source $s$.
\end{definition}

Coding over long block of symbols often improves the rate of
point-to-point codes. Similarly, increased efficiency may be expected
for network codes operating over a long block of source
symbols. Therefore, we also consider the asymptotic tradeoff between
source rates and link capacities.

\begin{definition}[Asymptotically admissible]
  A rate-capacity tuple $(\inputRate , \edgeRate)$ is
  \emph{asymptotically admissible} if there exists a sequence of
  zero-error network codes
  $\networkCoding^{(n)} = \{\edgeRV_f^{(n)},
  f\in\sessions\cup\edges\}$ and positive normalizing constants
  $\normal(n)$ such that
  \begin{align*}
    \lim_{n\to\infty} \frac{1}{\normal(n)} H\left(\edgeRV_e^{(n)}\right) \le
    \lim_{n\to\infty} \frac{1}{\normal(n)}\log |\edgeAlphabet_e^{(n)}|
    &\le   \edgeRate_e,  \quad\forall e \in \edges,\\
    \lim_{n\to\infty} \frac{1}{\normal(n)} H\left(\inputRV_s^{(n)}\right) =
    \lim_{n\to\infty} \frac{1}{\normal(n)} \log
    |\edgeAlphabet_s^{(n)}| &\ge \inputRate_s, \quad\forall s\in \sessions.
\end{align*}
\end{definition}

The above two definitions consider zero-error network codes. Relaxing
the requirement to allow arbitrarily small error probability prompts
the following definition.
\begin{definition}[Achievable rate-capacity tuple]\label{df:achievable}
  A rate-capacity tuple $(\inputRate , \edgeRate)$ is
  \emph{achievable} if there exists a sequence of network codes
  $\networkCoding^{(n)} \triangleq \{\edgeRV_f^{(n)},
  f\in\sessions\cup\edges\}$ and positive normalizing constants
  $\normal(n)$ such that
  \begin{align*}
    \lim_{n\to\infty} \frac{1}{\normal(n)} H\left(\edgeRV_e^{(n)}\right) \le
    \lim_{n\to\infty} \frac{1}{\normal(n)}\log |\edgeAlphabet_e^{(n)}|
    &\le
    \edgeRate_e,  \quad\forall e \in \edges,\\
    \lim_{n\to\infty} \frac{1}{\normal(n)}H\left(\inputRV_s^{(n)}\right) =
    \lim_{n\to\infty} \frac{1}{\normal(n)} \log
    |\edgeAlphabet_s^{(n)}| &\ge \inputRate_s, \quad\forall s\in \sessions, \\
    \lim_{n\to\infty} P_e\left(\networkCoding^{(n)}\right) =0.
\end{align*}
\end{definition}

Assuming that the underlying network and connection requirement are
known implicitly, the set of admissible, asymptotically admissible and
achievable rate-capacity tuples will be denoted $\Upsilon^0 ,
\Upsilon^\infty $ and $\Upsilon^\epsilon$ respectively.

The preceding definitions place no restriction on the class of network
codes under consideration. However, if a rate-capacity tuple is
admissible/asymptotically admissible/achievable using a network code
in a specific class $\codeclass$ (e.g. the class of linear network codes),
then that rate-capacity tuple is said to be admissible/asymptotically
admissible/achievable by network codes in $\codeclass$, and the
corresponding sets are denoted $\Upsilon^0_\codeclass  ,
\Upsilon^\infty_\codeclass$ and $\Upsilon^\epsilon_\codeclass$.

In this paper, we are interested in two special classes of network
codes, (i) linear network codes (with respect to an underlying
finite field $\field$) and (ii) abelian group network codes. The sets
of admissible/asymptotically admissible/achievable rate-capacity
tuples by linear network codes are respectively denoted by
$\Upsilon^0_{L(q)} , \Upsilon^\infty_{L(q)}$ and
$\Upsilon^\epsilon_{L(q)}$. Similarly, the set of
admissible/asymptotically admissible/achievable rate-capacity tuples
by abelian group network codes are respectively denoted by
$\Upsilon^0_\abelian , \Upsilon^\infty_\abelian$ and
$\Upsilon^\epsilon_\abelian$.

Discovering the hidden structure of these sets of rate-capacity tuples
is the key to understanding the tradeoff between source rates and edge
capacities. In the following, we list some basic structural properties of
$\Upsilon^0_\codeclass  , \Upsilon^\infty_\codeclass$ and
$\Upsilon^\epsilon_\codeclass$ when ${\codeclass}$ is either the class of
all network codes, linear network codes or abelian group network codes.

{\renewcommand{\theenumi}{P\arabic{enumi}}
\begin{enumerate}
\item\label{structure1} The sets $\Upsilon^0_\codeclass,
  \Upsilon^\infty_\codeclass$ and $\Upsilon^\epsilon_\codeclass$
  are closed under addition. In other words, if tuples $(\inputRate ,
  \edgeRate )$ and $(\inputRate^\prime , \edgeRate^\prime )$ are in
  $\Upsilon^0_\codeclass$ (or respectively in $\Upsilon^\infty
  _\codeclass$ and $\Upsilon^\epsilon _\codeclass$), then the
  element-wise addition of the two tuples will still be in the same
  set.
\item\label{structure2} $\Upsilon^\infty_\codeclass$ and
  $\Upsilon^\epsilon_\codeclass$ are closed convex cones, and
  $\con(\Upsilon^0_\codeclass) = \Upsilon^\infty_\codeclass$
  where $\con(\Upsilon^0_\codeclass)$ is the minimal closed convex
  cone containing $\Upsilon^0_\codeclass$.
\item\label{structure3} Admissibility implies asymptotic admissibility
  which further implies achievability, $\Upsilon^0_\codeclass \subseteq
  \Upsilon^\infty_\codeclass \subseteq \Upsilon^\epsilon_\codeclass$.
\end{enumerate}}

\section{Pseudo-variables and bounds}
\label{sec:bounds}
The sets of admissible/achievable rate-capacity tuples are difficult
to characterize explicitly. In fact, we will show later that finding
these sets is at least as hard as determining the set of entropy
functions $\Gamma^*$. Due to the difficulty of the problem, results on
characterizing the set of achievable rate-capacity tuples are quite
limited
\cite{Chan07capacity,Chan05capacity,Song.Yeung.ea03zero-error,Dougherty.Freiling.ea07matroids}. While
inner bounds and outer bounds constructed with entropic/almost
entropic functions exist \cite{Yeung02first}, these bounds are not
computable and hence are of limited practical use. The only known
computable outer bound is the Linear Programming (LP) bound, which is
constructed using polymatroids \cite{Yeung02first}. The remainder of
this section provides a brief review of these bounds. We use the
opportunity to introduce notation (differing slightly from the
original manuscripts), facilitating later discussion.

Let $\groundset$ be a nonempty finite set. Recall that
$\espace[\groundset]$ (or simply $\espace$) is a real euclidean space
which has $2^{|\groundset|}$ dimensions and coordinates indexed by the
set of all subsets of $\groundset$ and that $g(\emptyset)=0$ for all
$g\in \espace[\groundset]$.  Specifically, if $g\in\espace$, then its
coordinates will be denoted by $(g(\myalpha): \myalpha \subseteq
\groundset)$. We call $\groundset$ a ground set. Each $g\in\espace$
can also be viewed as a real-valued function $g:2^{\groundset}\mapsto
\Real$ defined on each subset of $\groundset$.

\begin{definition}[Polymatroid]\label{def:polymat}
  A function $g \in\espace[\groundset]$ is a \emph{polymatroid} if it satisfies
  \begin{align}
    g(\emptyset) &= 0 \label{poly:zero} \\
    g(\myalpha) &\ge g(\mybeta),\quad\text{if}\ \mybeta\subseteq \myalpha &&\text{non-decreasing} \label{poly:nondec}\\
    g(\myalpha) + g(\mybeta) &\ge g({\myalpha\cup\mybeta}) +
  g({\myalpha\cap\mybeta}) && \text{submodular} \label{poly:submod}
  \end{align}
\end{definition}
Note \eqref{poly:zero} and \eqref{poly:nondec} imply non-negativity of
a polymatroid.  Let $\groundset$ be a set of discrete random variables
with finite entropies. Note that $\groundset$ contains random
variables rather than indexes for a set of random variables.  This
induces a function $g\in\espace$ where $g(\myalpha)$ is the joint
entropy of the set of random variables
$\emptyset\neq\myalpha\subseteq\groundset$. Functions so-defined will
be called \emph{entropy functions}.

It is well-known that entropy functions are polymatroids over the
ground set $\groundset$. In fact, in the context of entropy functions,
the polymatroid axioms are completely equivalent to the basic
information inequalities (i.e. non-negativity of conditional mutual
information) \cite[p.  297]{Yeung02first}.  It is by now well-known
however that there are other information inequalities that are not
implied by the polymatroid axioms. The set of entropy functions is
denoted $\Gamma^*$, while the set of polymatroids is $\Gamma$.

While an entropy function takes a subset of random variables as
argument, a polymatroid $g$ more generally takes a subset of the
ground set $\groundset$ as argument, where the elements of
$\groundset$ may or may not be random variables. For simplicity, we
shall call the elements of the ground set of a polymatroid
\emph{pseudo-variables}. They differ from random variables in that
they do not necessarily take values, and there may be no associated
joint probability distribution function.

It must be emphasized that pseudo-variables are only defined in the
context of a polymatroid $g$ defined on the ground set $\groundset$. The
elements of $\groundset$ are not pseudo-variables by themselves in the
absence of an associated polymatroid.

Carrying these ideas further, we will call $g(\myalpha)$ the
\emph{pseudo-entropy} of the set of pseudo-variables $\myalpha$, and
$g$ is a \emph{pseudo-entropy function}. Treating pseudo-variables as
a set of basic objects associated with a polymatroid yields notational
simplification. For example, random variables are simply
pseudo-variables possessing a probability distribution such that their
pseudo-entropy function is the same as the entropy function. As such,
we extend the use of $H(\myalpha)$ to refer to the pseudo-entropy of a
set of pseudo-variables $\myalpha$.

\begin{definition}[Entropic function]
  A set of pseudo-variables (and its associated pseudo-entropy
  function) is called \emph{entropic} if its pseudo-entropy
  function is the same as an entropy function of a set of random
  variables.

  Similarly, a set of pseudo-variables (and their pseudo-entropy
  function) is called \emph{linear group characterizable} if its
  pseudo-entropy function is the same as an entropy function of a set
  of linear group characterizable random variables.
\end{definition}

The following two definitions generalize concepts of functional
dependence and independence to pseudo-variables.
\begin{definition}[Functional dependence]\label{def:function}
  Let $\groundset$ be a set of pseudo-variables. A pseudo-variable
  $X\in\groundset$ is said to be a \emph{function} of a set of
  pseudo-variables $\myalpha\subseteq\groundset$ if $H\left(\{X\} \cup
    \myalpha\right) = H\left(\myalpha\right)$. This relation will be
  denoted by $H(X|\myalpha) = 0$.
\end{definition}
\begin{definition}[Independence]
  Two subsets of pseudo-variables $\myalpha$ and $\mybeta$ are called
  \emph{independent} if $H(\myalpha \cup\mybeta) = H(\myalpha ) +
  H(\mybeta)$, and this relationship will be denoted by $\myalpha \indep
  \mybeta$.  Similarly, if $H(\bigcup_{j\in{\cal J}} \myalpha_j ) =
  \sum_{j\in{\cal J}} H(\myalpha_j)$, write $\indep_{j\in{\cal J}} \myalpha_j$.
\end{definition}

Clearly, these definitions are consistent with the usual ones used for
random variables.  The following bound re-states the linear
programming bound \cite[Section 15.6]{Yeung02first} in terms of
pseudo-variables.
\begin{definition}[LP bound]\label{def:LP}
  Given a network $\graph$ and a connection requirement
  $\multicastRequirement$, the LP bound is the set of rate-capacity
  tuples $(\inputRate , \edgeRate ) $ such that there exists a set of
  pseudo-variables $\{\inputRV_s : s\in\sessions, \edgeRV_e:
  e\in\edges \}$ satisfying the following ``connection constraint'':
  \begin{equation}\label{eqn:outerbd}
\begin{split}
   H\left(\edgeRV_e\mid \edgeRV_f : {\eht{f}{e}}\right) & = 0, \quad  e\in\edges\\
   H\left(\inputRV_s\mid \edgeRV_f : {\eht{f}{u}}\right) & = 0, \quad
   u\in\destinationLocation(s)\\
  \indep_{s\in\sessions} & \; \inputRV_s  \\
   H(\inputRV_s) & \ge \inputRate_s, \quad s\in\sessions \\
   H(\edgeRV_e) & \le \edgeRate_e, \quad e\in\edges.
\end{split}
\end{equation}
\end{definition}

Denote the set of rate-capacity tuples that satisfy the LP bound by
$\Upsilon_{LP}$. From \cite{Yeung02first} it is known that
$\Upsilon_{LP}\supseteq\Upsilon^\epsilon$.  It is interesting to
notice that the use of pseudo-variables gives a notational unification
of an inner bound and an outer bound given in \cite{Yeung02first} as
follows:
\begin{proposition}[Inner and Outer bounds]
  Given a network $\graph$ and a connection requirement
  $\multicastRequirement$, let $\Upsilon_{\text{in}}$ resp.
  $\Upsilon_{\text{out}}$ be the set of rate-capacity tuples
  $(\inputRate , \edgeRate ) $ such that there exists a set of
  \emph{entropic} resp. \emph{almost entropic} pseudo-variables
  $\{\inputRV_s : s\in\sessions, \edgeRV_e: e\in\edges \}$ satisfying~\eqref{eqn:outerbd}.
Then $\Upsilon_{\text{in}} \subseteq \Upsilon^\epsilon \subseteq
\Upsilon_{\text{out}} \subseteq \Upsilon_{LP}$.
\end{proposition}
\begin{proof}
The proof is straightforward by rewriting the bounds obtained in \cite{Yeung02first}.
\end{proof}

Similar to the LP bound, we define the following bound for abelian
group network codes (including linear network codes) as follows.
\begin{definition}[LP-Ingleton bound]\label{def:LPIngleton}
  Given a network $\graph$ and a connection requirement
  $\multicastRequirement$, the LP-Ingleton bound is the set of
  rate-capacity tuples $(\inputRate , \edgeRate ) $ such that there
  exists a set of pseudo-variables $\{\inputRV_s : s\in\sessions,
  \edgeRV_e: e\in\edges \}$ satisfying the Ingleton inequalities
  (\ref{eq:ingleton}) and the connection constraint
  (\ref{eqn:outerbd}).
\end{definition}

\begin{proposition}
  Denote the set of rate-capacity tuples that satisfy the LP-Ingleton bound by
  $\Upsilon_{LP,I}$. Then $\Upsilon_{LP,I}$ contains
  $\Upsilon^\epsilon_\abelian$.
\end{proposition}
\begin{proof}
  First notice that all source and link random variables of an abelian
  group network code must satisfy the Ingleton inequalities. The
  proposition then follows by using a similar argument as in
  \cite{Yeung02first} that proves
  $\Upsilon_{LP}\supseteq\Upsilon^\epsilon$.
\end{proof}

Since the LP and LP-Ingleton bounds are defined by intersections of
several linear half-spaces and hyperplanes, these bounds are
polyhedral. Together with the following duality results, this implies
that LP bounds are not generally tight (this is proved Section
\ref{sec:applicationA}).


\section{Entropy functions, network codes and duality}\label{sec:result}
Given a network, a connection requirement and a rate-capacity tuple,
the \emph{multicast problem} is to determine whether or not the
rate-capacity tuple is admissible or achievable (perhaps even
restricted to codes in a particular class). In this section, we
construct multicast problems from non-negative functions.  This
construction yields several dualities between properties of the
generating function and the solubility of the multicast problem. We
establish three main dualities.  The first duality relates entropy
functions and network codes. It can be paraphrased as follows.
\begin{quote}
  \emph{A function is quasi-uniform if and only if its induced
    rate-capacity tuple is admissible.}
\end{quote}
This is shown in Theorem \ref{thm:firstDuality}.  Theorem
\ref{thm:FirstDualityExtension} provides an extension which implies
\begin{quote}
  \emph{A function is almost entropic if and only if its induced
  rate-capacity tuple is achievable.}
\end{quote}
The second duality proves similar results for linear network codes.
\begin{quote}
\emph{An entropy function is linear group characterizable if and only if
  its induced rate-capacity tuple is admissible by linear network codes.}
\end{quote}
This is Theorem \ref{thm:SecondDuality}. Again, Theorem
\ref{thm:SecondDualityExtension} extends the result, relating almost
linear group characterizable functions and achievable rate-capacity
tuples with linear network codes.

The third duality, Theorem \ref{thm:LPbd}
relates polymatroids and the linear programming bound.
\begin{quote}
  \emph{A function is a polymatroid if and only if its induced
    rate-capacity tuple satisfies the LP bound.}
\end{quote}
We also give a partial result  for an extension to polymatroids that
also satisfy the Ingleton inequality.

Despite their apparent simplicity, these results leads to many
interesting corollaries: linear network codes (or more generally,
abelian group network codes) are suboptimal, the LP bound is not
tight, and in general the network coding capacity region is not a
polytope. These consequences will be described in more detail in
Section \ref{sec:applicationA}.

\subsection{Constructing multicast problems}
\label{sec:construction}
Let $h\in \espace[\N]$, be a given non-negative function over the
ground set $\N=\{1,2,\dots,N\}$. The proof for the main result relies
on the construction of a special network $\graph^\dagger$, a
connection requirement $\multicastRequirement^\dagger$ and a
rate-capacity tuple $\multicastProblem(h) \triangleq (\inputRate(h) ,
\edgeRate(h))$.

Figure \ref{fig:thenetwork} defines the network topology, connection
requirement and edge capacities. For convenience, the network is
divided into several subnetworks. To differentiate the roles of
network nodes, source nodes are indicated by open circles, destination
nodes are double circles, and intermediate nodes are solid circles. By
construction, each node takes only one role. The label beside a source
node is the input message available to that source node (this defines
the source location mapping $\sourceLocation$). The label beside a
receiver node indicates the desired source message to be reconstructed
at that destination node (this defines the destination location
mapping $\destinationLocation$). To simplify notation, each
capacitated edge is labeled with a pair of symbols denoting the edge
message (and corresponding random variable), and the edge
capacity. Unlabelled edges are assumed to be uncapacitated, or to have
a finite but sufficiently large capacity (such as $\sum_{\alpha}
h(\alpha)$) to losslessly forward all received messages.

The first part of the network, shown in Figure \ref{fig:part1},
contains the sources. There are $2^N-1$ independent sessions,
$\sessions = \left\{ \sRV_{[\alpha]} : \emptyset\neq\alpha\subseteq
  2^\N \right\}$\footnote{For simplicity, we use the same symbol to
  denote the index of a multicast session and the associated source
  random variable.}. The desired source rate associated with session
$\alpha$ is $h(\alpha)$.  Singletons $\{i\}\in 2^\N$ will be denoted
without brackets, e.g. $h(i)$ and $\sRV_{[i]}$.  There are $N$
specific edge messages that are of particular interest. Rather than
naming all edge variables $\edgeRV_e, e\in\edges$, we label these $N$
particular edge variables $V_j$, $j=1,\dots,N$. Remaining edge
variables will be labelled with generic symbols $W, W^\prime, W'', W^*$ and
$W^{**}$. Source $\sRV_{[\N]}$ generates the network coded messages $V_1,
V_2, \dots, V_N$ which are duplicated as required and forwarded to the
rest of the network. The remaining part of the network is divided into
subnetworks of three types, shown in Figures \ref{fig:type0},
\ref{fig:lowerbd} and \ref{fig:butterfly}.

\def\scalefactor{0.75}
\begin{figure}[htbp]
  \begin{center}
    \subfigure[The sources.\label{fig:part1}]
    {\includegraphics*[scale=\scalefactor]{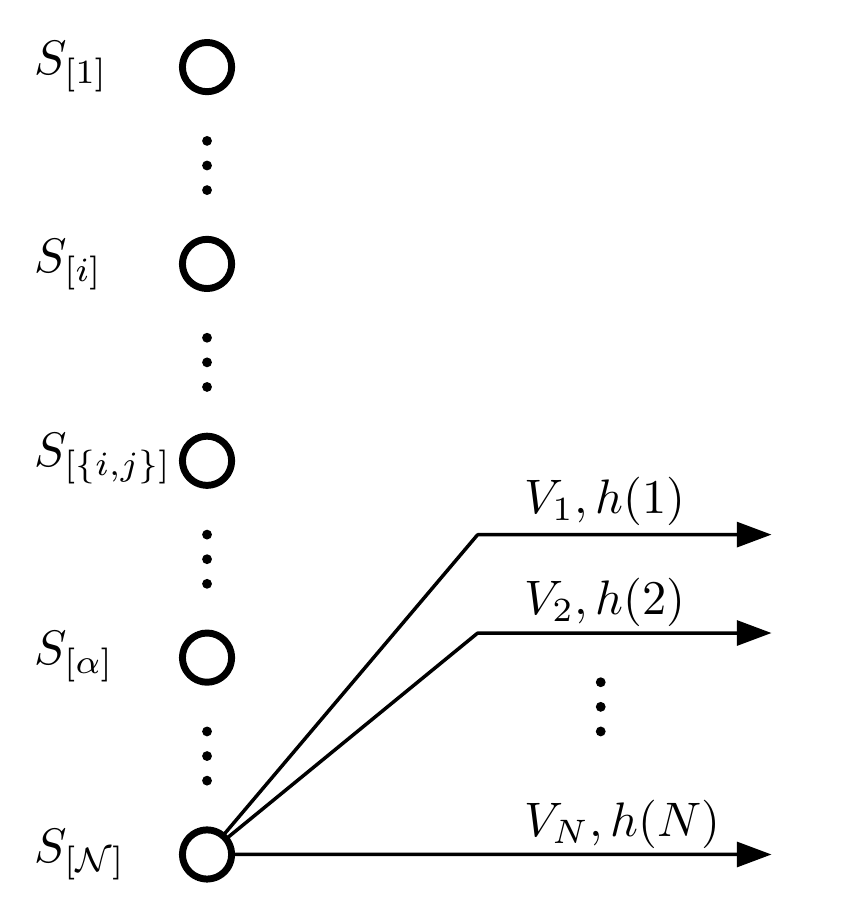}}
   \subfigure[Type 0 subnetworks\label{fig:type0}]
    {\includegraphics*[scale=\scalefactor]{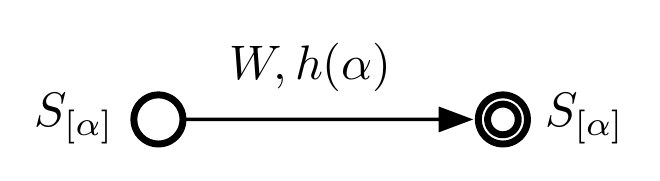}}
    \hspace{2mm}
    \subfigure[{Type 1 subnetworks}\label{fig:lowerbd}]
    {\includegraphics*[scale=\scalefactor]{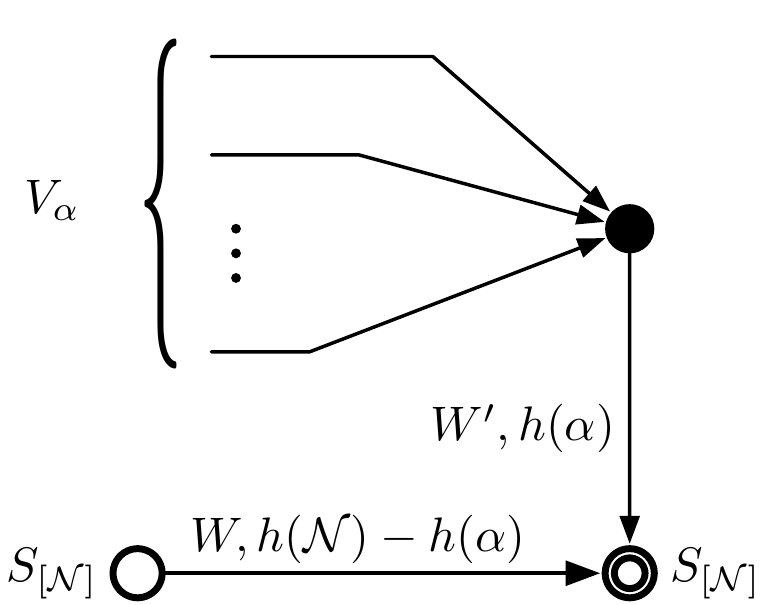}}

    \subfigure[{Type 2 subnetworks}\label{fig:butterfly}]
    {\includegraphics*[scale=\scalefactor]{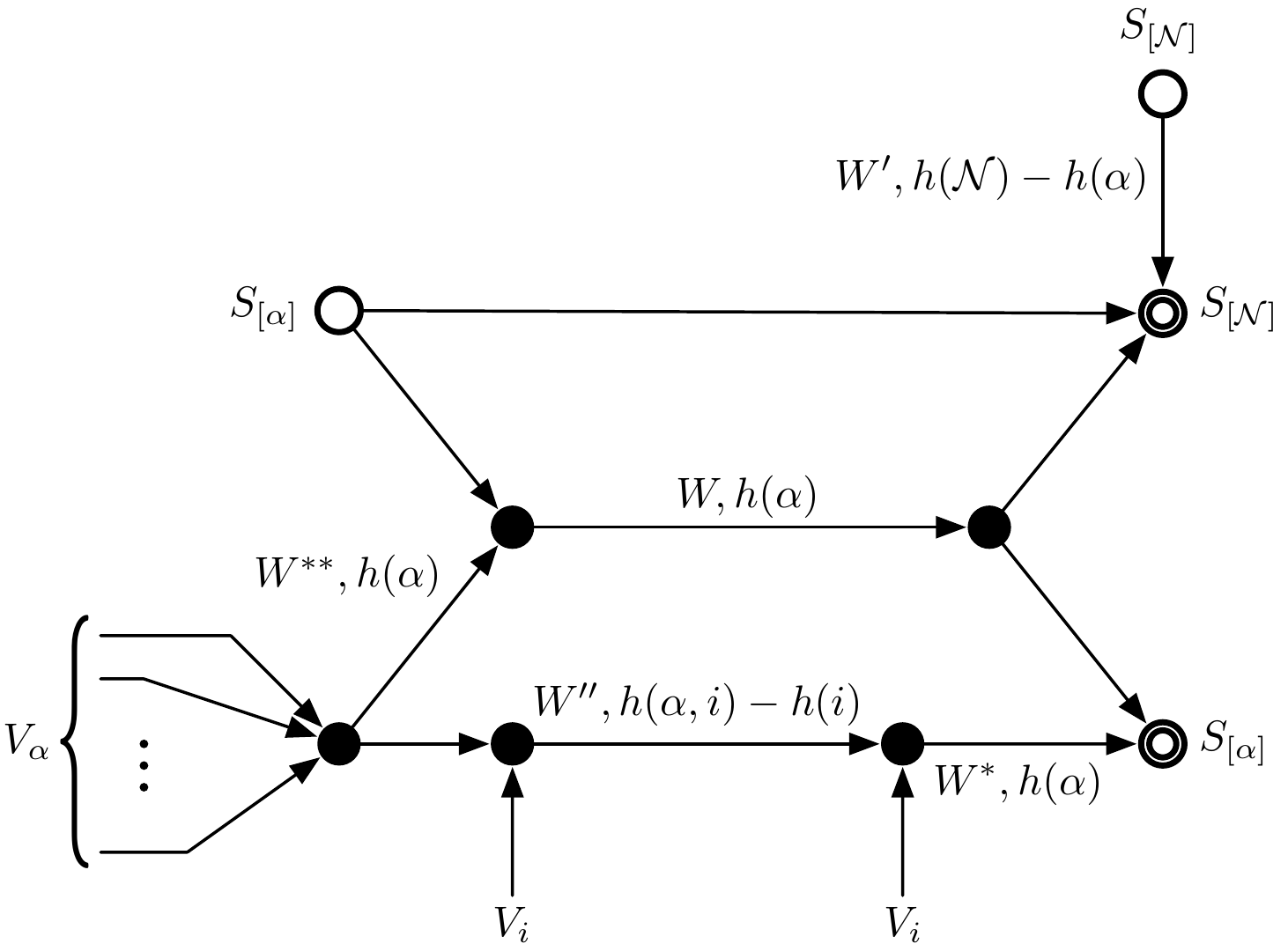}}
  \end{center}
  \caption{The network $\graph^\dagger$.}
  \label{fig:thenetwork}
\end{figure}

With reference to Figure \ref{fig:type0}, type 0 subnetworks connect a
single source to one receiver. There are $2^N-1$ type 0 subnetworks,
indexed by the choice of $\emptyset\neq\alpha\in2^\N$.

Referring to Figure \ref{fig:lowerbd}, there are $2^N-1$ type 1
subnetworks, one for each nonempty $\alpha\in 2^{\N}$. These
subnetworks introduce an edge of capacity $h(\N)-h(\alpha)$ between
source $\sRV_{[\N]}$ and a sink requiring $\sRV_{[\N]}$. There is an
intermediate node which has another $|\alpha|$ incident edges (from
Figure \ref{fig:part1}), carrying the messages
$V_\alpha=\{V_j,j\in\alpha\}$.  The intermediate node then has an edge
of capacity $h(\alpha)$ to the sink.

Finally, Figure \ref{fig:butterfly} shows the structure of the type 2
subnetworks. Type 2 subnetworks are indexed by a set $\alpha$, where
$\emptyset\neq\alpha\subset\N$ and an element $i\in\alpha,
i\not\in\N$. Each type 2 subnetwork connects two sources
$\sRV_{[\alpha]}$ and $\sRV_{[\N]}$ and two receivers respectively
requiring $\sRV_{[\alpha]}$ and $\sRV_{[\N]}$. In addition, there are
$|\alpha|+2$ other incident edges from Part 1 of the network, carrying
$V_\alpha$ and two copies of $V_i$. For notational simplicity, we have
written $h\left(\alpha\cup\{i\}\right)\defined h(\alpha,i)$.

So far, we have described a network $\graph^\dagger$, a connection
requirement $\multicastRequirement^\dagger$ and have assigned rates to
sources and capacities to links.  Clearly
$\multicastRequirement^\dagger$ depends only on $N$, and not in any
other way on $h$. Similarly, the topology of the network
$\graph^\dagger$ depends only on $N$. The choice of $h$ affects only
the source rates and edge capacities, which are collected into the
rate-capacity tuple $\multicastProblem(h)$. Also, we can assume without loss
of generality that  $\multicastProblem(h)$ is a linear function of $h$.

\begin{example}
  Figure \ref{fig:thenetworkNistwo} shows the topology of the network
  $\graph^\dagger$ when $N=2$. Edge labels are omitted for clarity.
\end{example}
\begin{figure}[htbp]
  \begin{center}
  \includegraphics*[scale=\scalefactor]{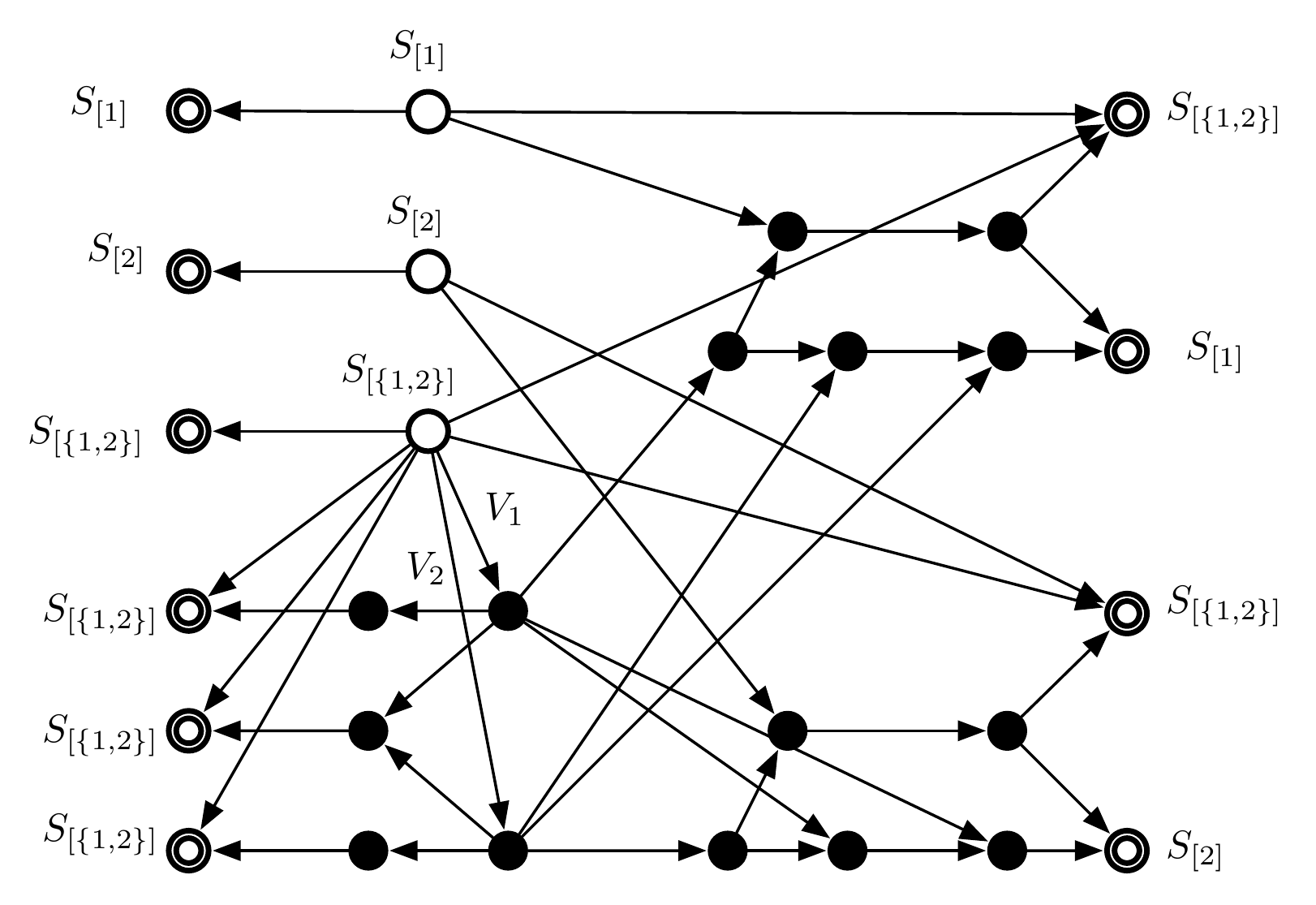}
  \end{center}
  \caption{The network $\graph^\dagger$ when $N=2$.}
  \label{fig:thenetworkNistwo}
\end{figure}

\subsection{First Duality: Entropy functions and network codes}
\label{sec:firstDuality}
\begin{theorem}\label{thm:firstDuality}
  Let $h$ be in $\espace[\N]$ for  $\N=\{1,2,\dots, N\}$. The
  induced rate-capacity tuple $\multicastProblem(h)$ is admissible on
  the network $\graph^\dagger$ and connection requirement
  $\multicastRequirement^\dagger$, if and only if $h$ is quasi-uniform, i.e.,
  \begin{equation*}
    h\in\Gamma^*_Q \iff \multicastProblem(h) \in \Upsilon^0.
  \end{equation*}
\end{theorem}
We begin with a proof of the only-if statement, i.e. starting with the
assumption of admissibility, we must demonstrate that the function is
quasi-uniform.  By Definition \ref{df:admissible}, admissibility of
$\multicastProblem(h)$ on $\graph^\dagger,
\multicastRequirement^\dagger$ requires existence of a zero-error
network code $\networkCoding$ with source messages $\sRV_{[\alpha]}$,
$\emptyset \neq \alpha\subseteq \N$ and a subset of its coded messages
$V_\N$ satisfying
\begin{align}
  H\left(\sRV_{[\alpha]}\right) &\ge h(\alpha), \quad \alpha \subseteq
  \N \label{thm:one:1}\\
  H\left(\sRV_{[\alpha]} : \alpha \subseteq \N\right) &= \sum_{\alpha
    \subseteq \N} H(\sRV_{[\alpha]}) \label{thm:one:2} \\
  H\left(V_i\right) &\le h(i), \quad i\in\N. \label{thm:one:3}
\end{align}
The remaining goal is to prove $H(V_\alpha) = h(\alpha)$ for every
$\alpha\subseteq\N$. To this end, we prove the following series of
Lemmas \ref{claim:0}--\ref{subclaim43}, each predicated on
admissibility of $\multicastProblem(h)$ on $\graph^\dagger,
\multicastRequirement^\dagger$.
\begin{lemma}\label{claim:0}
  $H\left(\sRV_{[\alpha]}\right) = h(\alpha)$ for all
  $\emptyset\neq\alpha\subseteq\N$.
\end{lemma}
\begin{proof}
  Consider the type 0 subnetworks of Figure
  \ref{fig:type0}. Admissibility implies that each receiver can
  correctly reconstruct its required source message. This is not
  possible unless $H(\sRV_{[\alpha]}) \le H(W) \le h(\alpha)$, which together
  with (\ref{thm:one:1}) proves the lemma.
\end{proof}

\begin{lemma}\label{claim:1}
  $ h(\alpha) \le H(V_\alpha)$ for all
  $\emptyset\neq\alpha\subseteq\N$.
\end{lemma}
\begin{proof}
  Consider type 1 subnetworks in Figure \ref{fig:lowerbd}. In order
  for the receiver to correctly determine the requested source message
  $\sRV_{[{\cal N}]}$, it must be true that $H(V_\alpha) + H(W) \ge
  H(\sRV_{[\N]})$. Furthermore, $H(W)\leq h(\N)-h(\alpha)$. Hence,
  \begin{align*}
    H(V_\alpha) + h({\cal N}) - h(\alpha) &\ge H(V_\alpha) + H(W) \\
    &\ge H(\sRV_{[{\cal N}]}) \\
    &\ge  h({\cal N}),
  \end{align*}
  where the last line follows from (\ref{thm:one:1}).  As a result,
  $H(V_\alpha)\ge h(\alpha)$.
\end{proof}

\begin{lemma}\label{claim:2}
  $H(V_j) = h(j)$ for all $j \in \N$.
\end{lemma}
\begin{proof}
  A direct consequence of Lemma \ref{claim:1} and (\ref{thm:one:3}).
\end{proof}

By Lemma \ref{claim:2} we have taken a small step towards our goal,
establishing $H(V_\alpha)=h(\alpha)$ for $|\alpha|=1$. Extension to
all $\alpha$ will be achieved by induction on $|\alpha|$. To this end,
the remaining lemmas take the hypothesis $H(V_\alpha)=h(\alpha)$ for
$|\alpha|=k<N$, and are proved in the context of type 2 subnetworks
indexed by $\alpha$ and an element $i\in\N$, $i\not\in\alpha$, as
shown in Figure \ref{fig:butterfly}.

\begin{lemma}
\label{subclaim41}
  In type 2 subnetworks, $W \indep \sRV_{[\alpha]}$. Furthermore, if
  $V_\alpha=h(\alpha)$, then $H(V_\alpha |  W, \sRV_{[\alpha]})=0$.
\end{lemma}
\begin{proof}
By (\ref{thm:one:2}), $\sRV_{[\alpha]}\indep\sRV_{[{\cal N}]}$ and hence
\begin{align*}
  H\left(\sRV_{[\alpha]}\right) + H\left(\sRV_{[\N]}\right) &=
  H\left(\sRV_{[\alpha]}, \sRV_{[\N]}\right) \\
  &\leq H\left(\sRV_{[\alpha]}, \sRV_{[{\cal N}]}, W, W'\right) \\
  &\nle{(i)}  H(W, \sRV_{[\alpha]}, W') \\
  &= H(W,\sRV_{[\alpha]}) + H(W'\mid W,\sRV_{[\alpha]}) \\
  &\nle{(ii)} H(W,\sRV_{[\alpha]}) + H(W')\\
  &\le H(W) + H(\sRV_{[\alpha]}) + H(W')\\
  &\nle{(iii)} h(\alpha) + H(\sRV_{[\alpha]}) +   H(W')  \\
  &\nle{(iv)} h(\alpha) + H(\sRV_{[\alpha]}) +  h({\cal N}) - h(\alpha)  \\
  &\nequal{(v)}  H(\sRV_{[\alpha]}) + H(\sRV_{[{\cal N}]}).
\end{align*}
The inequality $(i)$ follows from the fact that $\sRV_{[{\cal N}]}$ is
determined from $W, \sRV_{[\alpha]}, W^\prime$ at the upper receiver
in Figure \ref{fig:butterfly}. Inequality $(ii)$ is by discarding
conditioning (note that both $W$ and $W'$ depend on $\sRV_{[\N]}$, so
this is indeed only an inequality). Inequalities $(iii)$ and $(iv)$ follow from
the type 2 subnetwork capacity constraints,
\begin{align}
  H(W)&\leq h(\alpha) \label{eq:W}\\
  H(W')&\leq h(\N)-h(\alpha) \label{eq:Wprime}
\end{align}
and from Lemma \ref{claim:0}. Finally, $(v)$ is by Lemma
\ref{claim:0}.  Thus the series of inequalities is actually a series
of identities, and as a result,
\begin{align}
  H(W)&=h(\alpha) \label{eq:Walpha}\\
  H(W, \sRV_{[\alpha]}) &= H(W) +
  H(\sRV_{[\alpha]}) = 2h(\alpha) \label{eq:WS}
\end{align}
which proves $W\indep \sRV_{[\alpha]}$. Now consider
\begin{align*}
  H(V_\alpha | W, \sRV_{[\alpha]}) & = H(V_\alpha, W,
  \sRV_{[\alpha]}) - H( W, \sRV_{[\alpha]})\\
  &\nequal{(i)} H(V_\alpha,  \sRV_{[\alpha]}) - H( W, \sRV_{[\alpha]})\\
  &\leq H(V_\alpha) + H(\sRV_{[\alpha]}) - H( W, \sRV_{[\alpha]}) \\
  &\nequal{(ii)} H(V_\alpha) - h(\alpha) \\
  &=0\ \text{if}\ H(V_\alpha)=h(\alpha)
\end{align*}
where $(i)$ holds since $W$ is a function of $V_\alpha,
\sRV_{[\alpha]}$ and $(ii)$ is by (\ref{eq:Walpha}) and (\ref{eq:WS}).
\end{proof}

\begin{lemma}\label{subclaim42}
  In type 2 subnetworks, $H(W|V_\alpha, W^*) = H(W|W^*) = H(W)$, or
  equivalently, $I(W;V_\alpha, W^*)=0$.
\end{lemma}
\begin{proof}
  Recalling that $i\not\in\alpha\subset\N$,
  \begin{align*}
    H(W|V_\alpha, W^*) &\ge H(W | V_\alpha, W^*, V_i) \\
    &\nequal{(i)} H(W | V_\alpha, V_i) \\
    &\nequal{(ii)} H(W | V_\alpha, V_i) +
    H(\sRV_{[\alpha]} |V_\alpha, V_i, W ) \\
    &= H(W, \sRV_{[\alpha]} | V_\alpha, V_i) \\
    &\ge H( \sRV_{[\alpha]} | V_\alpha, V_i) \\
    &\nequal{(iii)} H(\sRV_{[\alpha]}) \\
    &\nequal{(iv)} h(\alpha) \\
    &\nge{(v)} H(W) \\
    &\ge  H(W|W^*) \\
    &\ge H(W|V_\alpha, W^*)
  \end{align*}
  where $(i)$ follows from the fact that $W^*$ is a function of
  $V_\alpha, V_i$, $(ii)$ follows from that $\sRV_{[\alpha]}$ can be
  reconstructed at the lower receiver, and $(iii)$ follows from
  independence of $\sRV_{[\alpha]}$ and $(V_\alpha, V_i )$, since by
  (\ref{thm:one:2}) $\sRV_{[\alpha]} \indep \sRV_{[\N]}$ and all the
  $V_j,j\in\N$ depend only on $\sRV_{[\N]}$. Finally, $(iv)$ is by
  Lemma \ref{claim:0}, $(v)$ is by the capacity constraint
  (\ref{eq:W}) and the remaining inequalities simply add extra
  conditioning. Thus the chain of inequalities is actually a chain of
  identities, the last three proving the lemma.
\end{proof}

\begin{lemma}\label{subclaim43}
  In type 2 subnetworks, assuming $H(V_\alpha)=h(\alpha)$,
  $H(W^*|V_\alpha) = H(V_\alpha | W^*) = 0$.
\end{lemma}
\begin{proof}
  \begin{align*}
    H(V_\alpha| W^*) &= H(V_\alpha| W^*, W) + I(V_\alpha ; W | W^*) \\
    &\nequal{(i)} H(V_\alpha| W^*, W)   \\
    &\le H(V_\alpha , \sRV_{[\alpha]}| W^*, W) \\
    &= H(V_\alpha | W^*, W, \sRV_{[\alpha]})  + H(\sRV_{[\alpha]}| W^*, W) \\
    &\nequal{(ii)} H(V_\alpha | W^*, W, \sRV_{[\alpha]}) \\
    &\leq H(V_\alpha |  W, \sRV_{[\alpha]})\\
    &\nequal{(iii)} 0.
  \end{align*}
  where $(i)$ follows from Lemma \ref{subclaim42}, $(ii)$ is
  because $\sRV_{[\alpha]}$ can be reconstructed at the lower
  receiver, and $(iii)$ is by Lemma \ref{subclaim41}, assuming
  $H(V_\alpha)=h(\alpha)$. Since conditional entropies are
  non-negative
  \begin{equation}\label{eq:alphaWstar}
    H(V_\alpha| W^*)=0.
  \end{equation}
  On the other hand,
  \begin{align*}
    H(W^*|V_\alpha) &= H(W^*, V_\alpha) - H(V_\alpha) \\
    &= H(W^*)+  H(V_\alpha | W^*) - H(V_\alpha) \\
    &\leq h(\alpha) - h(\alpha) = 0
  \end{align*}
  where the last inequality uses \eqref{eq:alphaWstar}, the type 2 subnetwork
  capacity bound $H(W^*)\leq h(\alpha)$ and the assumption
  $H(V_\alpha)=h(\alpha)$. Non-negativity of conditional entropy
  yields $H(W^*|V_\alpha)=0$.
\end{proof}

We are now ready to assemble the preceding lemmas into a proof for
the only-if part of Theorem \ref{thm:firstDuality}.
\begin{proof}[Proof: only-if part of Theorem \ref{thm:firstDuality}]
  The goal is to prove $H(V_\alpha) = h(\alpha)$ for all non-empty
  subsets $\alpha\subseteq\N$. This was already shown for $|\alpha| =
  1$ in Lemma \ref{claim:2}. Extension to all $\alpha$ will be
  achieved using induction. First, assume the hypothesis is true for
  all $\alpha\subset\N$ with $1\leq|\alpha| \le k < N$. For any $i\in\N$
  and $\alpha\subset\N$ such that $i\not\in\alpha$ and $|\alpha|=k$,
  consider the type 2 subnetwork of Figure \ref{fig:butterfly}. We
  must show that $H(V_\alpha, V_i) = h(\alpha \cup \{i\})\defined
  h(\alpha, i)$.  By Lemma \ref{claim:1} we already know that
  $H(V_\alpha, V_i) \geq h(\alpha, i)$. Therefore it remains only to
  prove $H(V_\alpha, V_i) \le h(\alpha, i)$. Now
  \begin{align*}
    H(V_i, V_\alpha) &\le H(V_i, V_\alpha, W^*) \\
    &\nequal{(i)} H(V_i,  W^*)\\
    &\leq H(V_i, W^*, W'') \\
    &\nequal{(ii)} H(V_i, W'') \\
    &\le H(V_i) + H(W^{\prime\prime})\\
    &\nle{(iii)} H(V_i) + h(\alpha,i) - h(i) \\
    &\nequal{(iv)} h(i)  + h(\alpha,i) - h(i) \\
    &= h(\alpha,i)
  \end{align*}
  where $(i)$ follows from Lemma \ref{subclaim43} (which holds under
  the induction hypothesis), $(ii)$ is due to the fact that $W^*$ is a
  function of $W'',V_i$ and $(iii)$ is from the subnetwork 2 capacity
  bound $H(W'')\leq h(\alpha,i) - h(i)$. Finally, $(iv)$ is by Lemma
  \ref{claim:2}.

  Up to this point, we have proved that $h$ is the entropy function of
  a set of random variables $\{V_1,\dots, V_N\}$. To show that $h$ is
  indeed quasi-uniform, it suffices to prove that for any subset
  $\alpha$ of $\N$, the set of random variables $V_\alpha$ is
  quasi-uniform.  Since we have just showed that
  $H(V_\alpha)=h(\alpha)$, if the receiver in the type 1 subnetwork
  can decode $\sRV_{[\N]}$, then $H(V_\alpha|W^\prime) =
  H(W^\prime|V_\alpha)=0$. Hence, $H(W^\prime)=h(\alpha)$. Now
  according to the link capacity constraint, $W^\prime$ is defined on
  an alphabet set of size $2^{h(\alpha)}$, and $W^\prime$ (and hence
  $V_\alpha$) must be quasi-uniform.
\end{proof}

It remains to prove the ``if'' statement in the theorem, i.e. to
show that quasi-uniform random variables imply admissibility.
\begin{proof}[Proof: if part of Theorem \ref{thm:firstDuality}]
  It suffices to show that one can construct a network code (defined
  by input variables, and message variables) meeting the connection
  requirement subject to the individual capacity constraint on each
  link.

  The construction for the input variables is simple. For any
  $\emptyset\neq\alpha\subseteq\N$, define $\sRV_{[\alpha]}$ to be a
  quasi-uniform random variable with entropy $h(\alpha)$. These input
  variables are also assumed to be independent. It remains to show
  that we can construct edge variables satisfying the capacity
  constraints, and which allow each receiver to reconstruct the
  requested messages perfectly.

  By the quasi-uniformity of $\sRV_{[\alpha]}$, it is clear that all
  receivers in type 0 subnetworks can reconstruct their requested
  message simply by having the source transmit the uncoded message,
  $W=\sRV_{[\alpha]}$.

  Let $\{V_j : j\in\N\}$ be a set of quasi-uniform random variables
  whose entropy function is $h$. Since $H(V_\N)=H(\sRV_{[\N]})$, there
  is a one-to-one mapping between $\supp(V_\N)$ and
  $\supp(\sRV_{[\N]})$. As they are both quasi-uniform, $\sRV_{[\N]}$
  and $(V_j : j\in\N)$ can be regarded as the same.

  For type 1 networks, by quasi-uniformity of $V_\alpha$, one can send
  $V_\alpha$ unencoded as $W^\prime$. Then the receivers see
  $V_\alpha$ and an auxiliary message $W$ defined on a sample space of
  size at most $2^{h(\N)-h(\alpha)}$. Reconstructing $\sRV_{[\N]}$ at
  the receiver is equivalent to reconstructing $V_{\N \backslash \alpha}$
  at the receiver.

  By the quasi-uniformity of $\sRV_{[\alpha]}$ and Lemma
  \ref{lemm:qucoding}, $V_{\N \backslash \alpha}$ can be compressed to
  a symbol $W$ of size $2^{h(\N)-h(\alpha)}$ such that $V_{\N
    \backslash \alpha}$ can be losslessly reconstructed from $W$ and
  $V_\alpha$.

  It remains to verify that receivers in type 2 subnetworks can
  reconstruct all requested messages.  Recall that both
  $\sRV_{[\alpha]}$ and $V_\alpha$ are quasi-uniform. Assume without
  loss of generality that their supports are $\{0, 1, 2, \dots,
  2^{h(\alpha)} -1 \}$. Then we can define $W\defined V_\alpha +
  \sRV_{[\alpha]} \mod 2^{h(\alpha)}$. It is easy to verify the
  following properties:
  \begin{align}
    H\left(W\mid V_\alpha, \sRV_{[\alpha]}\right) &=
    H\left(\sRV_{[\alpha]} \mid W, V_\alpha\right) =  H\left(V_\alpha
      \mid W,\sRV_{[\alpha]}\right) = 0, \label{eq:wprop1}\\
    \log|\supp(W)| &= h(\alpha).
  \end{align}
  By \eqref{eq:wprop1}, the upper receiver can correctly reconstruct
  $V_\alpha$ from $\sRV_{[\alpha]}$ and $W$. Using a similar
  compression scheme as used in type 1 subnetworks, source
  $\sRV_{[\N]}$ is compressed to $h(\N)-h(\alpha)$ bits, allowing
  lossless reconstruction of $\sRV_{[\N]}$ at the upper receiver.

  On the other hand, it is easy to see that $\{V_\alpha, V_i\}$ is
  quasi-uniform. Hence $V_\alpha$ can be compressed into
  $W^{\prime\prime}$ with a support of size $ |\supp(W^{\prime\prime})
  | = 2^{h(\alpha, i) - h(i)}$ such that $V_\alpha$ can be
  reconstructed by using $W^{\prime\prime}$ and $V_i$. As a result,
  $W^*$ may be transmitted as $V_\alpha$ without any encoding. The
  lower receiver can then recover $\sRV_{[\alpha]}$ from $V_\alpha$
  and $W$.

  Since all receivers can reconstruct their requested source messages
  with properly constructed message random variables satisfying the
  capacity constraints, the rate-capacity tuple $\multicastProblem(h)$
  is admissible.
\end{proof}

\begin{definition}
  A polymatroid $h$ is called \emph{almost entropic} if there exists a
  sequence of entropic pseudo-entropy functions $h^{(k)}$ and positive
  constants $r(k)$ such that $\lim_{k\to\infty} h^{(k)}/r(k) = h$.
\end{definition}

As $\bar{\Gamma}^*$ is a closed and convex cone
\cite{Yeung97framework}, the set of all almost entropic functions is
$\bar{\Gamma}^*$.  Theorem \ref{thm:firstDuality} establishes a
duality, or equivalence between the quasi-uniformity of $h$ and
admissibility of $\multicastProblem(h)$. The following theorem extends
this result to a duality between almost entropic $h$ and
asymptotically admissible (and achievable) $\multicastProblem(h)$.

\begin{theorem}\label{thm:FirstDualityExtension}
  Let $h\in\espace[\N]$ for $\N=\{1,2,\dots, N\}$ and let
  $\multicastProblem(h)$ be an induced rate-capacity tuple. Then we
  have,
  \begin{equation*}
    h \in \bar{\Gamma}^* \iff \multicastProblem(h) \in \Upsilon^\infty
    \iff  \multicastProblem(h) \in \Upsilon^\epsilon .
  \end{equation*}
  In other words, the rate-capacity tuple $\multicastProblem(h)$ is
  asymptotically admissible (or achievable) on the network
  $\graph^\dagger$ and connection requirement
  $\multicastRequirement^\dagger$ if and only if $h$ is almost
  entropic.
\end{theorem}
\begin{proof}
  Suppose that $h$ is almost entropic. We will first show that
  $\multicastProblem(h) \in \Upsilon^\infty$.  By
  \cite{Chan.Yeung02relation, Chan1998}, one can construct a sequence
  of quasi-uniform entropic functions $h^{(n)}$ and normalizing
  constants $r(n)$ that $\lim_{n\to\infty} h^{(n)}(\alpha)/{r(n)} =
  h(\alpha)$. By Theorem \ref{thm:firstDuality}, each
  $\multicastProblem(h^{(n)})$ is admissible. By property
  \ref{structure2}, the set $\Upsilon^\infty$ of asymptotically
  admissible rate-capacity tuples is a closed and convex cone and
  hence $\multicastProblem(h) \in \Upsilon^\infty$.

  Clearly, $\multicastProblem(h) \in \Upsilon^\infty$ implies that $\multicastProblem(h) \in \Upsilon^\epsilon$.
  It remains to show that $\multicastProblem(h)$ is achievable implying  that $h$ is almost entropic.
  Suppose that $\multicastProblem(h) \in \Upsilon^\epsilon$.
  According to
  Definition \ref{df:achievable}, one can construct a sequence of
  normalizing constants $r(n)$ and network codes
  $\networkCoding^{(n)}$ with source messages
  $\{\sRV_{[\alpha]}^{(k)}, \alpha\subseteq \N \}$ and edge messages
  $V_N^{(k)}$ such that\footnote{By the Bolzano-Wierstrass Theorem
    which says that any sequence in a closed and bounded interval has
    a convergent subsequence, we can safely assume that
    $\lim_{k\to\infty} \frac{1}{r(k)} H( \sRV_{[\alpha]}^{(k)} ,
    V_\beta^{(k)})$ exists for any nonempty subsets $\alpha, \beta$ of
    $\N$.}
  \begin{align}
   \lim_{k\to\infty} \frac{1}{r(n)} H\left(\sRV_{[\alpha]}^{(n)}\right) &\ge
    h(\alpha) \\
   \lim_{k\to\infty} \frac{1}{r(n)} H\left(V_i^{(n)}\right) &\le h(i)\\
   \lim_{n\to\infty} P_e\left(\networkCoding^{(n)}\right) &= 0.
  \end{align}
  For each value of the sequence index $n$, consider the network
  $\graph^\dagger$ and connection requirement
  $\multicastRequirement^\dagger$ of Figure \ref{fig:thenetwork} with
  sources $\sessions=\left\{\sRV^{(n)}_{[\alpha]},
    \emptyset\neq\alpha\in 2^\N \right\}$ and edge messages
  $V_\N^{(n)}$.  By the Fano inequality, the entropy of any source
  $s\in\sessions$ conditioned on the edge variables incident to any
  node in $\destinationLocation(s)$ can be made as small as desired by
  increasing $n$.  Following a similar procedure as in the proof for
  Theorem \ref{thm:firstDuality}, it can be proved that for any
  non-empty subset
  $\emptyset\neq\alpha\subseteq\N$, $$\lim_{k\to\infty} \frac{1}{r(n)}
  H\left(V_\alpha^{(k)}\right) = h(\alpha).$$ In other words, $h$ is
  almost entropic.
 \end{proof}

 \subsection{Second Duality: Linear group characterizable functions
   and linear network codes}
\label{sec:secondDuality}
 The first duality shows that $h$ is quasi-uniform (almost entropic)
 if and only if $\multicastProblem(h)$ is admissible (achievable). We
 will now prove a similar result, restricting the network codes to be
 linear.
\begin{theorem}\label{thm:SecondDuality}
  Let $h\in\espace[\N]$ for  $\N=\{1,2,\dots, N\}$. The
  induced rate-capacity tuple $\multicastProblem(h)$ is admissible
  using linear network codes on the network $\graph^\dagger$ and
  connection requirement $\multicastRequirement^\dagger$, if and only
  if $h$ is linear group characterizable, i.e.,
   \begin{equation*}
     h\in \Gamma^*_{L(q)} \iff \multicastProblem(h) \in
     \Upsilon^0_{L(q)}
   \end{equation*}
\end{theorem}
\begin{proof}[Proof: only-if part of Theorem \ref{thm:SecondDuality}]
  The proof of the only-if part is very similar to the one given in
  Theorem \ref{thm:firstDuality}. Suppose that $\multicastProblem(h)
  \in \Upsilon^0_{L(q)}$, i.e., it is admissible using a
  linear network code $\networkCoding $ on the network
  $\graph^\dagger$ and connection requirement
  $\multicastRequirement^\dagger$.  By Proposition
  \ref{prop:linearCodeAndLinearChar}, the set of induced source and
  link random variables by $\networkCoding $ is linear group
  characterizable. Using the same argument as in the proof for Theorem
  \ref{thm:firstDuality}, $h$ is the entropy function of a subset of
  these linear group characterizable random variables. Hence, $h$ is
  linear group characterizable.

  In fact, using the same argument, we can show that if the induced
  rate-capacity tuple $\multicastProblem(h)$ is admissible using
  abelian network codes on the network $\graph^\dagger$ and connection
  requirement $\multicastRequirement^\dagger$, then $h$ is abelian
  group characterizable.
\end{proof}

Before we prove the if part of Theorem \ref{thm:SecondDuality}, we
need the following lemma which serves a similar role as Lemma
\ref{lemm:qucoding} in the proof of Theorem \ref{thm:firstDuality} by
justifying the feasibility of certain ``compression'' scheme.
\begin{lemma}\label{lemma:sideinfonetwork}
  Consider a special case of the network depicted in Figure
  \ref{fig:sideinfo} where the left node receives $T_1(a)$ and $T_2(a)$
  as inputs, where $T_1$ and $T_2$ are two linear functions defined on
  a vector space ${\bf A}$ over $\field$.  Let the kernels of $T_1$
  and $T_2$ be respectively ${\bf B}_1$ and ${\bf B}_2$. Then, there
  exists a linear function $W$ of $T_1(a)$ and $T_2(a)$ such that (1)
  $T_1(a)$ is uniquely determined from $W$ and $T_2(a)$, and (2) $W$
  takes at most $q^{\dim {\bf B}_2 - \dim {\bf B}_1 \cap {\bf B}_2}$
  different values.
\end{lemma}

\begin{proof}
  From ${\bf B}_1$ and ${\bf B}_2$, we can construct three subspaces
  ${\bf W}_1$, ${\bf W}_2$ and ${\bf W}_0$ such that $$\dim {\bf W}_0
  + \dim {\bf W}_1 + \dim {\bf W}_2 + \dim{\bf B}_1 \cap {\bf B}_2
  =\dim {\bf A}$$ and that for each $i=1,2$, the subspace ${\bf B}_i$
  is equal to the linear span of ${\bf W}_i$ and $ {\bf B}_1 \cap {\bf
    B}_2 $. Hence any $a\in {\bf A}$ can be written uniquely as
  $a=a_0+a_1+a_2+b$ where $a_i\in {\bf W}_i$ for $i=1,2,3$ and $b\in
  {\bf B}_1 \cap {\bf B}_2$.

Since $\kernel(T_1)={\bf B}_1$, we have $T_1(a_0+a_1+a_2+b) =
T_1(a_2)+T_1(b)$. Furthermore, one can easily construct a linear
function $T_1^*$ such that $T_1^*(T_1(a)) = (a_2, b)$. Similarly,
there exists a linear function $T_2^*$ such that $T_2^*(T_2(a)) =
(a_1, b)$.

To compute $T_1(a)$ at node 2, it suffices to compute $a_2$ as $b$ can
be computed directly from $T_2(a)$. A simple counting argument shows
that $a_2$ lies in a vector subspace of dimension $\dim {\bf B}_2 -
\dim {\bf B}_1 \cap {\bf B}_2$. Therefore, we can set
$W=a_2$ over the network and it takes at most
$q^{\dim {\bf B}_2 - \dim {\bf B}_1 \cap {\bf B}_2}$ different values.
\end{proof}

Now we may continue our proof for Theorem \ref{thm:SecondDuality}.
\begin{proof}[Proof: if part of Theorem \ref{thm:SecondDuality}]
  To prove the direct part of Theorem \ref{thm:SecondDuality}, we need
  to show that if $h$ is linear group characterizable, then one can
  construct a linear network code (defined by the induced source and
  link random variables) meeting the connection requirement subject to
  the individual capacity constraint on each link.

  Suppose that $h$ is linear group characterizable by a vector space
  ${\bf V}$ and its subspaces ${\bf V}_1,\dots, {\bf V}_N$, defined
  over a field $\field$. Assume without loss of generality that the
  subspaces intersect only at the zero vector, $\bigcap_{j=1}^N{\bf
    V}_j= \{{\bf 0}\}$. As such, $h(\N) = \log q \cdot (\dim {\bf V})$
  and for any $\alpha \subseteq \N$, we have $h(\alpha) =\log q \cdot
  (\dim {\bf V} - \dim \bigcap_{j\in\alpha} {\bf V}_j)$.

  For $j=1,\dots, N$, construct linear functions $f_j$ over ${\bf V}$
  such that $\kernel(f_j)={\bf V}_j$. The source random variable
  $\sRV_{[\N]}$ is uniformly distributed over ${\bf V}$ such that the
  link symbols transmitted in Figure \ref{fig:part1} are
  $V_j=f_j(\sRV_{[\N]})$. For any other
  $\emptyset\neq\alpha\subset\N$, define $\sRV_{[\alpha]}$ to be a
  random variable, uniformly distributed over a vector space of
  dimension $\log_q 2 \cdot h(\alpha)$ (hence,
  $H(\sRV_{[\alpha]})=h(\alpha)$). All these source random variables
  are assumed to be independent.

  Up to this point, we have described how source and link random
  variables are defined in Figure \ref{fig:part1}. It remains to show
  that we can construct a linear network code, consisting of a set of
  link random variables which are linear functions of the incident
  source/link random variables, satisfying the capacity constraints,
  and which allow each receiver to reconstruct the requested messages
  perfectly.

 For type 0 subnetworks, all receivers can reconstruct
their requested message simply by having the source transmit the
uncoded message, $W=\sRV_{[\alpha]}$.  Clearly, the associated link
random variables in these subnetworks are linear functions of the
incident ones and meet the capacity constraint.

For type 1 subnetworks, let $W^\prime = (V_i\where i\in\alpha) =
(f_i(\sRV_{[\N]}) \where i\in\alpha)$, which depends linearly on
$\sRV_{[\N]}$. Note that $(f_i(a) : i\in\alpha) = {\bf 0} $ if and
only if $ f_i(a)= {\bf 0}$ for all $i \in\alpha$, or equivalently,
when $a\in \bigcap_{i\in\alpha} {\bf V}_i$. By the rank-nullity
theorem, $W^\prime$ can take at most $|{\bf V}| /
|\bigcap_{i\in\alpha} {\bf V}_i|$ different values. We can thus treat
$W^\prime$ as a vector in space of dimension $\dim {\bf V} - \dim
\bigcap_{i\in\alpha} {\bf V}_i$.

As a result, the subnetwork can now be treated as a special case of
Lemma \ref{lemma:sideinfonetwork} such that $T_1(a) = a$ and $T_2(a) =
(f_i(a) : i \in \alpha)$. The dimensions of the kernels of $T_1$ and
$T_2$ are respectively $0$ and $\dim \bigcap_{i\in\alpha} {\bf
  V}_i$. By Lemma \ref{lemma:sideinfonetwork}, the required rate is
thus $\log q \cdot ( \dim \bigcap_{i\in\alpha} {\bf V}_i) =
h(\N)-h(\alpha)$.

Similarly, for type 2 subnetworks, let $W^{**} = (f_i(\sRV_{[\N]}) :
i\in\alpha)$.  As before, we can treat $W^{**}$ as a vector of length
$\dim {\bf V} - \dim \bigcap_{i\in\alpha} {\bf V}_i $. Similarly,
$\sRV_{[\alpha]}$ can also be regarded as a vector of the same length.
We can therefore define $W$ by vector addition,
$W=\sRV_{[\alpha]}+W^{**}$.  Consequently, the receiver in the upper
branch can reconstruct $V_\alpha$ by subtracting $\sRV_{[\alpha]}$
from $W$. As before, one can find $W^\prime$ as a linear function of
$\sRV_{[\N]}$ and this function allows $\sRV_{[\N]}$ to be
reconstructed from $W^\prime$ and $V_\alpha$.

For the lower branch, we can identify a special case of Figure
\ref{fig:sideinfo} with $T_1(a) = V_\alpha$ and $T_2(a) = V_i$. One can
construct $W^{\prime\prime}$ such that (1) $W^{\prime\prime}$ is a
linear function of $T_1(a) $ and $T_2(a)$, (2) the kernel
$\kernel(T_1) = \bigcap_{j\in\alpha} {\bf V}_j $ and
$\kernel(T_2)={\bf V}_i$, and (3) the rate required is $\dim
\bigcap_{j\in\alpha} {\bf V}_j - \dim {\bf V}_i \bigcap_{j\in\alpha}
{\bf V}_j$.  Therefore, we can reconstruct $V_\alpha$ from $W^{\prime\prime}$
and $T_2(a)$ where $T_1(a) = V_\alpha$. Again, treating $V_\alpha$ as
a vector of length $\dim {\bf V} - \dim \bigcap_{i\in\alpha} {\bf V}_i
$, the receiver at the lower branch can reconstruct $\sRV_{[\alpha]}$
by subtracting $V_\alpha$ from $W$.
\end{proof}

So far, we have proved that $h$ is linear group characterizable if and
only if the rate-capacity tuple $\multicastProblem(h)$ is admissible
with a linear network code.  As before, we can further generalize the
result to include the case when $h$ is almost linear group
characterizable according to the following definition.
\begin{definition}\label{def:almostlineargroupchar}
  A polymatroid $h$ is called \emph{almost linear group
    characterizable} if there exists a sequence of linear group
  characterizable entropy functions $h^{(k)}$ and positive constants
  $r(k)$ such that $\lim_{k\to\infty} h^{(k)}/r(k) = h$.
\end{definition}
It is easy to prove that the set of all almost linear group
characterizable polymatroids is $\con(\Gamma^*_{L(q)})$, the
minimal closed and convex cone containing $\Gamma^*_{L(q)}$.

\begin{theorem}\label{thm:SecondDualityExtension}
  Let $h\in\espace[\N]$ for $\N=\{1,2,\dots, N\}$ and let
  $\multicastProblem(h)$ be an induced rate-capacity tuple. Then we
  have
   \begin{equation*}
    h \in \con(\Gamma^*_{L(q)}) \iff \multicastProblem(h) \in \Upsilon^\infty_{L(q)} \iff  \multicastProblem(h) \in \Upsilon^\epsilon_{L(q)} .
  \end{equation*}
  In other words, the rate-capacity tuple $\multicastProblem(h)$ is
  asymptotically admissible (or achievable) by linear network codes on
  the network $\graph^\dagger$ and connection requirement
  $\multicastRequirement^\dagger$ if and only if $h$ is is almost
  linear group characterizable.
  \end{theorem}
\begin{proof}
  Suppose that $h \in \con(\Gamma^*_{L(q)})$.
  By Definition \ref{def:almostlineargroupchar}, one can
  construct a sequence of linear group characterizable entropy
  functions $h^{(k)}$ and positive constants $r(k)$ such that
  $\lim_{k\to\infty} h^{(k)}/r(k) = h$.  By Theorem
  \ref{thm:SecondDuality}, each $\multicastProblem(h^{(n)})$ is
  admissible by linear network codes. By property \ref{structure2},
  the set $\Upsilon^\infty_{L(q)}$ of asymptotically
  admissible rate-capacity tuples is a closed and convex cone and
  hence $ \multicastProblem(h) \in \Upsilon^\infty_{L(q)}$.

  Clearly, $ \multicastProblem(h) \in \Upsilon^\infty_{L(q)}$ implies that
  $\multicastProblem(h) \in \Upsilon^\epsilon_{L(q)}$. It remains to prove that
  $\multicastProblem(h) \in \Upsilon^\epsilon_{L(q)}$ implies $h \in \con(\Gamma^*_{L(q)})$.

  Suppose that $\multicastProblem(h)$ is achievable by linear network
  codes.  Then one can construct a sequence of normalizing constants
  $r(n)$ and linear network codes $\networkCoding^{(n)}$ with source
  messages $(\sRV_{[\alpha]}^{(k)}, \alpha\subseteq {\cal N} )$ and
  edge messages $(V_j^{(k)}, j \in {\cal N})$ such that
  \begin{align}
   \lim_{k\to\infty} \frac{1}{r(n)} H\left(\sRV_{[\alpha]}^{(n)}\right) &\ge
    h(\alpha) \\
   \lim_{k\to\infty} \frac{1}{r(n)} H\left(V_j^{(n)}\right) &\le h(j)\\
   \lim_{n\to\infty} P_e\left(\networkCoding^{(n)}\right) &= 0.
  \end{align}
  Similar to the proof given in Theorem   \ref{thm:FirstDualityExtension}, it can be proved that for any
  non-empty subset $\emptyset\neq\alpha\subseteq\N$,
  $\lim_{k\to\infty} \frac{1}{r(n)} H\left(V_\alpha^{(k)}\right) =
  h(\alpha).$ In addition, as $(V_j^{(k)}, j \in {\cal N})$ is linear
  group characterizable, $h$ is almost linear group characterizable.
\end{proof}

\subsection{Third Duality: Polymatroids and the LP bound}
\label{sec:thirdDuality}
Theorem \ref{thm:FirstDualityExtension} provides a duality between
entropy functions and network codes, namely that a function
$h\in\espace[\N]$ is almost entropic if and only if
$\multicastProblem(h)$ is achievable on $\graph^\dagger$,
$\multicastRequirement^\dagger$. As the set of almost entropic
functions $\bar{\Gamma}^*$ has no explicit characterization for four
or more variables, the sets of admissible or achievable rate-capacity
tuples are unknown. Therefore computable bounds such as the linear
programming bound are of great interest.

Let $\Gamma$ be the set of all polymatroids. Definition \ref{def:LP}
writes the LP bound in terms of constraints on pseudo-variables. The
following theorem provides a direct generalization of the ideas of the
previous sections to pseudo-variables.
\begin{theorem}\label{thm:LPbd}
  Suppose $h\in\espace[\N]$. A rate-capacity tuple $(\inputRate(h),
  \edgeRate(h))$ satisfies the LP bound if and only if $h$ is a
  polymatroid,
  \begin{equation*}
    h\in\Gamma \iff \multicastProblem(h)\in\Upsilon_{LP}.
  \end{equation*}
\end{theorem}
\begin{proof}
  The ``only if'' part of the proof is a direct generalization of the
  proof of Theorem \ref{thm:firstDuality}.  Suppose $(\inputRate(h) ,
  \edgeRate(h) )$ satisfies the LP bound. By Definition \ref{def:LP}
  there exists a set of pseudo-variables satisfying the set of
  (in)equalities in (\ref{eqn:outerbd}). In particular, there are
  pseudo-variables $\{\sRV_{[\alpha]}, \emptyset \neq \alpha\subseteq
  \N\}$ and $V_\N$ such that
  \begin{align}
    H(\sRV_{[\alpha]}) &\ge h(\alpha),  \quad \alpha\subseteq\N, \\
    H(\sRV_{[\alpha]} : \alpha\subseteq\N) &= \sum_{\alpha \subseteq
      \N} H(\sRV_{[\alpha]}) \\
    H(V_i) &\le h(i).
  \end{align}
  Following the same steps as in the proof for Theorem \ref{thm:firstDuality}
  (translating random variables to pseudo-variables), shows that $h$ is
  the pseudo-entropy function of $V_\N$. Hence, $h$ is a polymatroid.

  To prove the direct part, suppose $h$ is a polymatroid over the
  ground set $\groundset = \{V_1, V_2, \dots, V_N\}$ (i.e. $h$ is the
  pseudo-entropy function of $V_\N$). We must exhibit a set of
  pseudo-variables satisfying the set of (in)equalities
  (\ref{eqn:outerbd}).  Whereas the proof for Theorem
  \ref{thm:firstDuality} constructs auxiliary random variables via
  data compression, we need to show how to analogously adhere
  auxiliary pseudo-variables $W, W''$ etc. to the set of
  pseudo-variables $V_\N$. In contrast to random variables, we cannot
  rely on coding theorems, or other probabilistic constructions that
  assume the existence of an underlying probability
  distribution. Nevertheless, it is possible to adhere
  pseudo-variables. This is accomplished in Appendix \ref{app:LPbd},
  where proof of the direct part is also completed.
\end{proof}

\subsection{Fourth Duality: Ingleton polymatroids and the LP bound for
  linear codes?}
\label{sec:fourthDuality}
Finally, we can consider rate-capacity tuples which satisfy the
LP-Ingleton bound of Definition \ref{def:LPIngleton}. The following
theorem establishes a relation to Ingleton polymatroids (i.e., a
polymatroid satisfying Ingleton inequalities). This is shown in one
direction only. Let $\Gamma_{LP,I}$ be the set of all Ingleton polymatroids.
\begin{theorem}\label{thm:LPIngletonbd}
  Suppose $h\in\espace[\N]$. If a rate-capacity tuple $(\inputRate(h),
  \edgeRate(h))$ satisfies the LP bound for linear codes, then $h$ is
  an Ingleton polymatroid, i.e.,
  \begin{equation*}
     \multicastProblem(h)\in\Upsilon_{LP,I} \Rightarrow h\in\Gamma_{LP,I}.
  \end{equation*}
\end{theorem}
\begin{proof}
  Suppose $(\inputRate(h), \edgeRate(h) )$ satisfies the LP-Ingleton
  bound. By Definition \ref{def:LPIngleton} there exists a set of Ingleton
  pseudo-variables satisfying the set of (in)equalities in
  (\ref{eqn:outerbd}). In particular, there are pseudo-variables
  $\{\sRV_{[\alpha]}, \emptyset \neq \alpha\subseteq \N\}$ and $V_\N$
  such that
  \begin{align}
    H(\sRV_{[\alpha]}) &\ge h(\alpha),  \quad \alpha\subseteq\N, \\
    H(\sRV_{[\alpha]} : \alpha\subseteq\N) &= \sum_{\alpha \subseteq
      \N} H(\sRV_{[\alpha]}) \\
    H(V_i) &\le h(i).
  \end{align}
  Following the same steps as in the proof for Theorem \ref{thm:firstDuality}
  (translating random variables to pseudo-variables), shows that $h$
  is the pseudo-entropy function of $V_\N$. Hence, $h$ is an Ingleton
  polymatroid.
\end{proof}

We conjecture that the converse of the fourth duality should also
hold. In fact, it can be proved that if the converse fails to hold,
then there exists a polymatroid satisfying Ingleton inequalities but
which is not almost linear group characterizable. Therefore
determination of whether the converse of the fourth duality holds is a
very interesting open question.

\section{Implications}\label{sec:applicationA}
The results of Section \ref{sec:result} while interesting in their own
right, have several consequential applications. First, in Section
\ref{sec:multicast} we consider implications to the determination of
the network coding capacity region (in the absence of any restriction
on the class of network codes). Secondly, we discuss the sub-optimality
of linear network codes in Section \ref{sec:implications}.

\subsection{The capacity region}\label{sec:multicast}
\begin{implication}[Hardness of a multicast problem]
  Determination of the set of achievable source rate-link capacity
  tuples $\Upsilon^\epsilon$ is at least as hard as the problem of
  determining the set of all almost entropic functions.

  Similarly, determination of the set of source rate-link capacity
  tuples achieved by linear network codes
  $\Upsilon^\epsilon_{L(q)}$ is at least as hard as the
  problem of determining the set of all almost linear group
  characterizable entropy functions.
\end{implication}
\begin{proof}
  By Theorem \ref{thm:FirstDualityExtension}, a polymatroid $h$ is
  almost entropic (and almost linear group characterizable) if and
  only if the induced rate-capacity tuple $(\inputRate(h),
  \edgeRate(h) )$ is achievable (with linear network codes). In other
  words, the problem of determining the set of all almost entropic
  (and almost linear group characterizable) functions can be reduced
  to the solubility of a corresponding multicast problem.
\end{proof}

In \cite{Dougherty.Freiling.ea07matroids}, a network, called the V\'amos
network, was constructed from the V\'amos matroid. This was later used to
prove that the LP bound is not tight and the bound can be
tightened by applying a non-Shannon information inequality proved
in \cite{Zhang.Yeung98characterization}.

In the following, we will use the duality results obtained in Section
\ref{sec:result} to provide another proof for the looseness of LP
bound.
\begin{implication}[Looseness of LP bound]
  The LP outer bound can be tightened by any non-Shannon information
  inequality.
\end{implication}
\begin{proof}
  Theorem \ref{thm:LPbd} shows that the rate-capacity tuple
  $(\inputRate(h), \edgeRate(h))$ is in the LP bound if $h$ is a
  polymatroid. Yet, Theorem \ref{thm:FirstDualityExtension} proves
  that $(\inputRate(h), \edgeRate(h) )$ is achievable if and only if
  $h$ is almost entropic.  Consider the function $h$ defined as
  follows \cite{Zhang.Yeung98characterization}:
  \begin{align*}
    h(1) &= h(2) = h(3) = (4)  = 2a > 0 \\
    h(1, 2) &= 3a \\
    h(3, 4) &= 4a  \\
    h(1, 3) &=   h(1, 4) =   h(2, 3) =  h(2, 4) = 3a  \\
    h(i, j, k) &= 4a = h(1, 2, 3, 4), \mbox{ $\forall$ distinct $i,j,k$}.
  \end{align*}
  It can be verified directly that $h\in\Gamma_4$. However, the
  non-Shannon information inequality obtained in
  \cite{Zhang.Yeung98characterization} shows that
  $h\not\in\bar{\Gamma}_4^*$.  While the rate-capacity tuple
  $\multicastProblem(h)$ satisfies the LP bound, it is not achievable,
  as it is not almost entropic.
\end{proof}
Using the same argument, any non-Shannon information inequality
\cite{Zhang.Yeung98characterization,Dougherty.Freiling.ea06six,Matus07infinitely}
will remove some polymatroids which are not almost entropic. The
corresponding tuples in the LP bound will not be achievable. In other
words, \emph{any set of non-Shannon information inequalities can be
  used to tighten the LP bound}.

In fact, together with the fact that $\bar\Gamma^*$ is not a
polyhedron when the number of random variables is at least
four~\cite{Matus07infinitely}, our duality results lead to very
interesting consequences.

First, we show that the set of achievable rate-capacity tuples is not
a polyhedron in general. Second, the LP bound is not only loose, but
it remains loose even when tightened via application of any finite
number of linear non-Shannon information inequalities.

\begin{proposition}\label{prop:matus}
  The set of almost entropic functions is not a polytope.
\end{proposition}
\begin{proof}[Proof sketch] The following is a sketch of the proof
  given by Mat\'{u}\v{s} \cite{Matus07infinitely}.  Mat\'{u}\v{s}
  constructed a convergent sequence of entropic functions $g_t\to g_0$
  with one-side tangent $\dot{g}_{0+}\defined\lim_{t\to 0^+} (g_t -
  g_0)/t$.  Clearly, if $\bar{\Gamma}^*_n$ is polyhedral, there exists
  $\epsilon>0$ such that $g_0 + \epsilon
  \dot{g}_{0+}\in\bar{\Gamma}^*_n$. This was shown not to be the case,
  since $g_0 + \epsilon \dot{g}_{0+} $ violates some of the
  information inequalities proved in
  \cite{Matus07infinitely}. Therefore, $\bar{\Gamma}^*_n$ is not
  polyhedral. Furthermore, there are infinitely many information
  inequalities.
  \end{proof}

\begin{implication}[Set of achievable rate-capacity tuples]\label{thm:capacitynotpolytope}
  The sets of achievable rate-capacity tuples $\Upsilon^\infty$ and
  $\Upsilon^\epsilon$ for the network $ \graph^\dagger$ and connection
  requirement $\multicastRequirement^\dagger $ are not polytopes (when
  $N\ge 4$).
\end{implication}
\begin{proof}
  Consider the sequence $g_t\to g_0$ from the proof of Proposition
  \ref{prop:matus}.  By Theorem \ref{thm:FirstDualityExtension},
  $\multicastProblem(g_t)$ and $\multicastProblem(g_0)$ are
  asymptotically admissible. As $\multicastProblem(h)$ is a linear
  function of $h$, we have
\begin{align}
 \dot{T} \triangleq \lim_{t\to 0^+} (\multicastProblem(g_t) - \multicastProblem(g_0))/t = \multicastProblem(\dot{g}_{0+}).
\end{align}
For any $\epsilon>0$,
\begin{align}
\multicastProblem(g_0) + \epsilon \dot{T} = \multicastProblem(g_0  + \epsilon \dot{g}_{0+} ).
\end{align}
As $ g_0 + \epsilon \dot{g}_{0+} $ is not almost entropic,
$\multicastProblem(g_0) + \epsilon \dot{T}$ is not achievable. In
other words, $\Upsilon^\infty$ and $\Upsilon^\epsilon$ are not
polytope.
\end{proof}

Now the LP bound is a polytope, while the capacity region is
not. Furthermore, the introduction of any finite number of additional
linear inequalities in the LP bound simply results in another
polytope. Hence
\begin{implication}[Looseness of polyhedral bounds]
  The LP bound is not tight. Furthermore, any finite number of linear
  information inequalities cannot tighten the LP bound $\Upsilon_{LP}$
  to the set of achievable rate-capacity tuples $\Upsilon^\epsilon$.
  In fact, any polyhedral outer bound for $\Upsilon^\epsilon$ is not
  tight.
\end{implication}
\begin{proof}
  A direct consequence of Theorem \ref{thm:capacitynotpolytope} and
  Proposition \ref{prop:matus}.
\end{proof}

\subsection{Suboptimality of linear network codes}\label{sec:implications}
As discussed in Section \ref{sec:algebraic}, it may be practically
desirable to use network codes with nice algebraic properties that
simplify encoding and decoding operations. Most algebraic network
codes considered in the literature are linear, and these were shown in
\cite{Li.Yeung.ea03linear} to be optimal for single session multicast.

Since the appearance of \cite{Li.Yeung.ea03linear}, it has been an open
question as to whether linear network codes are in general
optimal. This question was recently answered in the negative by
Dougherty et. al \cite{Dougherty.Freiling.ea05insufficiency}. Their
proof constructs a special network containing two subnetworks such
that the base fields required for optimality by each of the
subnetworks have different characteristics, establishing a
contradiction.

The following provides an alternative proof using a completely
different approach, making use of the duality between entropy
functions and achievability established in Section \ref{sec:result}.
The proof is an immediate consequence of the duality results and that
some entropic functions are not almost linear group characterizable.

\begin{implication}[Suboptimality of linear network codes]
  There is a network and a connection requirement such that the use of
  abelian network codes is suboptimal, including linear network codes,
  $R$--module codes, and time-sharing of such.
\end{implication}
\begin{proof}
  Consider a set of four random variables $U_1, U_2,
  U_3, U_4$ constructed using the projective plane described in
  \cite{Zhang.Yeung98characterization}. The entropy function of these
  random variables is
  \begin{align*}
     h(1) &= h(2) = h(3) = (4)  = \log 13  \\
     h(1, 2) &= \log 6 + \log 13 \\
     h(3, 4) &= \log 13 + \log 12  \\
     h(1, 3) &=   h(1, 4) =   h(2, 3) =  h(2, 4) = \log 13 + \log 4  \\
     h(i, j, k) &= \log 13 + \log 12  = h(1, 2, 3, 4), \mbox{ $\forall$ distinct $i,j,k$}.
  \end{align*}
  Since $h$ is the entropy function of a set of random variables,
  $\multicastProblem(h)$ is achievable, by Theorem
  \ref{thm:FirstDualityExtension}. Since $h$ does not satisfy the
  Ingleton inequality
  \begin{multline}
      h(1 , 2) + h(1, 3) + h(1 , 4) +  h(2 , 3) + h(2 , 4)
    \ge  \\
      h(1 ) + h( 2) + h(3 , 4) +   h(1 , 2 , 3)
      + h(1 , 2 ,  4),
  \end{multline}
  $h$ is not almost linear group characterizable. By Theorem
  \ref{thm:SecondDualityExtension}, $\multicastProblem(h)$ is not
  achievable by linear network codes.
\end{proof}

\begin{implication}[Suboptimality of abelian group network codes]
  There is a network and a multicast requirement for which
  abelian  codes are (asymptotically) suboptimal.
\end{implication}
\begin{proof}
  All abelian group characterizable entropy function must satisfy the
  Ingleton inequality. The corollary then follows.
\end{proof}

\section{Conclusion}\label{sec:conclusion}
Entropy functions and network coding are already closely connected,
through the network coding capacity region which is expressed in terms
of $\Gamma^*$. The main results of this paper, summarized in
Figure~\ref{fig:summary}, further strengthens this
connection. Figure~\ref{fig:summary} shows the inclusion
relationships of the various sets of interest, as well as the
implications between set membership of $h$ and
$\multicastProblem(h)$ established by the theorems. Each arrow is
labeled by the Theorem number which establishes the relation. Note
that the relation of $\con(\Gamma^*_{L(q)})$ to sets other than
$\Gamma^*_{L(q)}$ shown in Figure~\ref{fig:summary1} is unknown, hence
the linear code relationships are shown separately in
Figure~\ref{fig:summary2}.
\begin{figure}[htbp]
  \centering
  \subfigure[\label{fig:summary1}]{\includegraphics[scale=0.8]{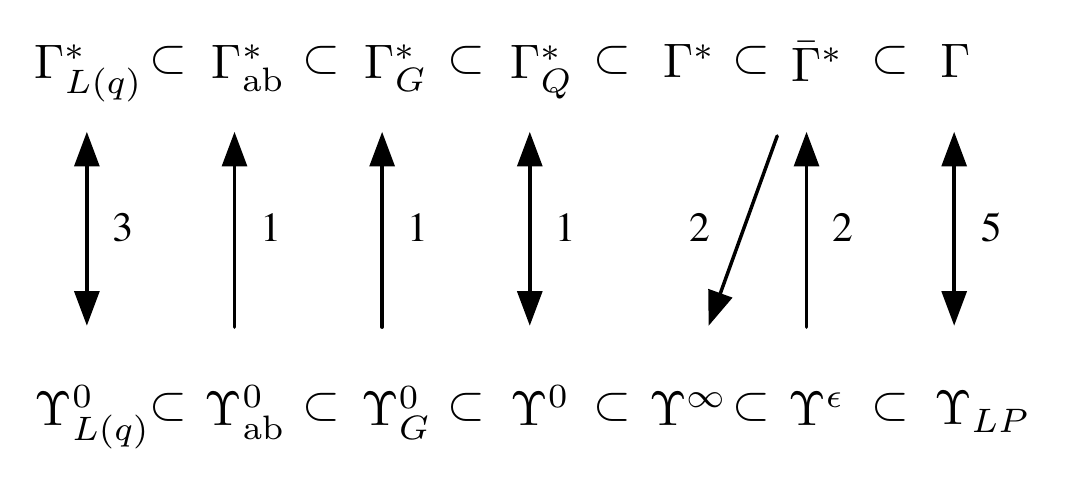}}
\hspace{5mm}
  \subfigure[Linear codes.\label{fig:summary2}]{\includegraphics[scale=0.8]{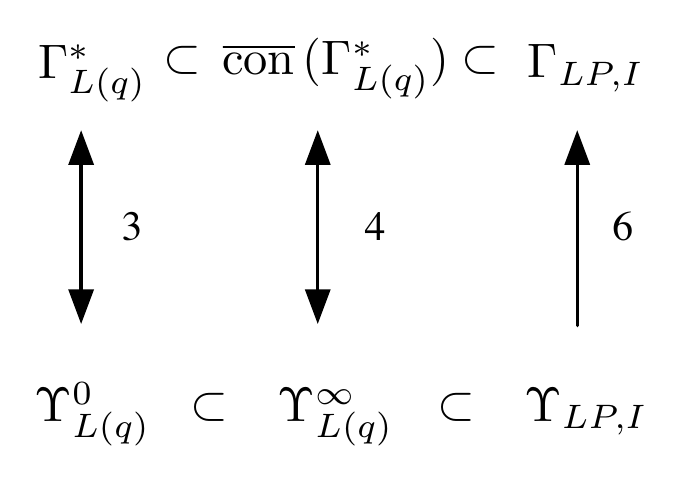}}
  \caption{Summary of the duality results.}
  \label{fig:summary}
\end{figure}

Given a non-negative real function $g$ whose domain
consists of all non-empty subsets of $N$ random variables, we have
provided a construction for a network and a connection requirement
such that a rate-capacity tuple is achievable if and only if $g$ is
almost entropic (i.e. satisfies every information inequality). The
network topology depends only on the number of random variables, and
not on the function $g$, which affects the construction only through
the assignment of source rates and link capacities.

An extension of this result shows that a rate-capacity tuple for the
constructed multicast problem is achievable by linear network codes if
and only if the entropy function $g$ is almost linear group
characterizable.  A further extension shows that the induced
rate-capacity tuple satisfies the linear programming bound if and only
if the function $g$ is a polymatroid (i.e. satisfies all Shannon-type
inequalities). This extension is obtained using the concept of
pseudo-variables, which replace random variables in the domain of
$g$. These pseudo-variables are abstract objects that do not take any
values, and are not associated with any probability distribution. The
key is that polymatroids defined over set of pseudo-variables behave
very similar to entropy functions, except that they lie in $\Gamma$
rather than $\Gamma^*$. This definition of pseudo-variables is not
just a matter of terminology. It is a non-trivial matter to generalize
notions of extension and adhesion of random variables (which rely on
the existence of a probability distribution) to pseudo-variables. We
provided some examples of such extensions and adhesions, which leaves
the proof of the main theorem intact under a substitution of
pseudo-variables for random variables. We anticipate that this concept
of pseudo-variables, and their differences from random variables, may
yet bear more fruit in uncovering the structure of $\Gamma^*$

The seemingly simple duality between entropy vectors and network codes
has a number of powerful implications. It renders the problems of
network code solubility is at least as hard as determination of
$\bar{\Gamma}^*$.  We also obtain alternate proofs that the LP bound
is not tight, and that non-Shannon inequalities such as the
Zhang-Yeung inequality indeed tighten the LP bound. However no
additional finite number of inequalities can improve the LP bound to
the capacity region. Finally, we have proved the suboptimality of
abelian network codes, including linear codes, $R$-module codes and
any scheme that time-shares between such codes. The duality result
also provides a tool to compare different classes of network
codes. Rather than comparing the codes directly, one can now compare
the sets of entropy functions induced by the codes.

\section*{Acknowledgement}
This work was supported in part by the Australian Government under ARC
grant DP0557310, and by the Defence Science and Technology
Organisation under contracts 4500485167 and 4500550654.

\appendices

\section{Proof for Converse of Theorem \ref{thm:LPbd}}
\label{app:LPbd}

Before we prove the direct part of Theorem \ref{thm:LPbd}, we will
prove some intermediate results which show how to \emph{extend} sets
of pseudo-variables (build new pseudo-variables from old ones), and
how to \emph{adhere} additional pseudo-variables to a given set of
pseudo-variables (consistently join two sets of
pseudo-variables). These results are provided in Section
\ref{sec:adherence}. The proof of Theorem \ref{thm:LPbd} follows in
Section \ref{sec:proofLPbd}.

\subsection{Adhesion and extension for pseudo-variables}
\label{sec:adherence}
For random variables, adhesion or extension is facilitated by the
existence of an underlying probability distribution. For example,
consider two sets of random variables $ \groundset = \{X,U\}$ and
$\groundset^* = \{X,W\}$ with respective underlying distributions
$P_{XU}$ and $P^*_{XW}$. Suppose that the marginals over $X$ coincide,
$P_X=P^*_X$.  We can then easily adhere $P_{XU}$ and $P^*_{XW}$ to
obtain a new distribution $Q_{XUW}$ such that its marginals over
$\groundset$ and $\groundset^*$ coincide, $Q_{XU}=P_{XU}$ and
$Q_{XW}=P^*_{XW}$. One possibility is $Q_{XUV}=P_{XU} P^*_{XW}/P_X$.
In general, for any sets of random variables $ \groundset$ and
$\groundset^*$ with respective distributions $P$ and $P^*$ coinciding
on $ \groundset \cap \groundset^*$, we can construct a new
distribution over $\groundset \cup \groundset^*$ such that its
marginals over $\groundset$ and $\groundset^*$ are $P$ and $P^*$.
Clearly, the entropy function for $\groundset \cup \groundset^*$ is an
extension of those belonging to $\groundset$ and $ \groundset^*$.

Consider another simple example. Let $\myalpha\subset\groundset$ be a
subset of the random variables $\groundset$. Then we can define a new
random variable $W \triangleq \myalpha$. By doing so, we have
constructed a new variable, and extended both the distribution and
entropy function. Clearly there are various ways to adhere or extend
sets of random variables. Doing this for pseudo-variables is not so
straightforward. The following results provide several adhesion and
extension methods for pseudo-variables.

\begin{lemma}[Functional extension]\label{lemm:functionadhesion}
  Let $\groundset$ be a set of pseudo-variables. For any given
  $\myalpha\subseteq\groundset$, one can adhere a new pseudo-variable
  $Y$ to $\groundset$ such that $H(Y|\myalpha) = H(\myalpha|Y) =0$.
  In other words, there exists a polymatroid $g$ over $\groundset \cup
  \{Y\}$ satisfying
  \begin{align}
    g(\mybeta) &= H(\mybeta) \quad\forall \mybeta\subseteq\groundset \label{eq:fa:1}\\
    g(Y) &= g(\myalpha) = g(\{Y\} \cup \myalpha).\label{eq:fa:2}
  \end{align}
\end{lemma}
\begin{proof}
  Define $g$ over $\groundset\cup \{Y \} $ such that for all
  $\mybeta\subseteq\groundset$,
  \begin{align}\label{eqn:functionadhesion}
    g(\mybeta)   = H(\mybeta)    \mbox{ and }
    g(\{Y\} \cup\mybeta)  = H(\mybeta \cup \myalpha).
\end{align}
It is straightforward to show  that $g$ is a polymatroid satisfying
\eqref{eq:fa:1} and \eqref{eq:fa:2}.
\end{proof}
In light of Definition \ref{def:function}, we shall refer to
(\ref{eqn:functionadhesion}) as functional extension and denote the
new variable as $J_\myalpha$. Clearly, any subset of pseudo-variables
in $\myalpha$ is a function of $J_\myalpha$.

\begin{lemma}[Sum extension]\label{lem:sumextension}
  Let $\{X,Y\}$ be a set of pseudo-variables such that $H(X) = H(Y)$
  and $X\perp Y$. Then one can adhere a new pseudo-variable $Z$ to
  $\{X,Y\}$ such that $H(Z)=H(X)$ and $H(Z|X,Y)=H(X|Y,Z)= H(Y|X,Z)=
  0$.
\end{lemma}
\begin{proof}
  Let $g$ be the pseudo-entropy function for $\{X,Y\}$. Extend $g$
  such that $g(Z) = g(X)$ and $g(X,Z) = g(Y,Z) = g(X,Y,Z) =
  g(X,Y)$. The resulting extended $g$ is still a polymatroid.
\end{proof}
Lemma \ref{lem:sumextension} shows that for any independent
pseudo-variables $X$ and $Y$ of equal pseudo-entropies, one can
construct a pseudo-variable $Z$, denoted $Z = X\oplus Y$ such that its
pseudo-entropy is the same as $X$ and $Y$, and any single
pseudo-variable is a function of the two others. Structurally, this
mimics the modulo-2 addition of two i.i.d binary random variables.

\begin{lemma}[SW extension]\label{lemm:SWadhesion}
  Let $\{X,Y\}$ be two pseudo-variables. Then one can adhere a new
  pseudo-variable $Z$ to $\{X,Y\}$ such that
  \begin{align*}
    H(Z) &= H(X|Y), \\
    H(X|Z,Y) &= 0, \\
    H(Z|X) &= 0.
  \end{align*}
\end{lemma}
\begin{proof}
  Let $g$ be the pseudo-entropy of $\{X,Y\}$ and extend it as follows:
  $g(Z)=g(X,Y) - g(Y)$, $g(Z,Y)=g(X,Y,Z)=g(X,Y)$, and $g(X,Z) =
  g(X)$. The resulting extended $g$ is still a polymatroid.
\end{proof}
Lemma \ref{lemm:SWadhesion} shows that starting with pseudo-variables
$X,Y$, one can construct another pseudo-variable $Z$ with
pseudo-entropy $H(X,Y) - H(Y)$ such that $X$ is a function of $Y,Z$
and $Z$ is a function of $X$.  For simplicity, we use the symbol
$J_{X|Y}$ to denote the new pseudo-variable $Z$.

Lemmas \ref{lemm:functionadhesion}--\ref{lemm:SWadhesion} show that
sets of pseudo-variables can be explicitly extended to obtain new
pseudo-variables. In the following, we study adhesion of existing sets
of pseudo-variables.

\begin{lemma}[Independent adhesion]\label{lemm:indpendentadhesion}
  Let $\groundset$ and $\groundset^*$ be two disjoint sets of
  pseudo-variables. Then they can adhere to each other
  \emph{independently} such that for any $\myalpha \subseteq
  \groundset \cup \groundset^*$,
  \begin{equation}\label{eqn:independentAdhesion}
    H(\myalpha) = H(\myalpha \cap  \groundset ) + H(\myalpha \cap
    \groundset^* ).
  \end{equation}
\end{lemma}
\begin{proof}
  Let $g$ and $g^*$ be the pseudo-entropies of $\myalpha$ and
  $\myalpha^*$, and for each $\myalpha \subseteq \groundset \cup
  \groundset^*$ set $g(\myalpha) = g(\myalpha \cap \groundset ) +
  g^*(\myalpha \cap \groundset^* )$. It can be verified that
  $g$ is a polymatroid.
\end{proof}
Any subsets $\myalpha \subseteq \groundset$ and $\mybeta \subseteq
\groundset^*$ are independent, $\myalpha\indep\mybeta$ under the
independent adhesion of $\groundset$ and $\groundset^*$ in Lemma
\ref{lemm:indpendentadhesion}. Before we continue with more
complicated adhesions, we need the following proposition from
\cite{Matus07adhesivity}.
\begin{proposition}
  Let $\groundset$ and $\groundset^*$ be two sets of pseudo-variables
  coinciding over $\groundset^\prime \defined \groundset \cap
  \groundset^*$, i.e. for all $\myalpha\subseteq \groundset^\prime $,
  the pseudo-entropy of $\myalpha$ is the same with respect
  $\groundset$ and $\groundset^*$. Further, suppose
  \begin{align}\label{eqn:flat}
    \Delta(\myalpha,\mybeta) \ge \Delta(\groundset^\prime \cap
    \myalpha, \groundset^\prime \cap \mybeta),
  \end{align}
  for all flats\footnote{A subset $\myalpha$ of the ground set
    $\groundset$ is a \emph{flat} if $H(\myalpha^\prime) >
    H(\myalpha)$ for all proper supersets $\myalpha^\prime$ containing
    $\myalpha$.}  $\myalpha,\mybeta$ of $\groundset$ where
  $\Delta(\myalpha,\mybeta) \defined H(\myalpha) + H(\mybeta) -
  H(\myalpha\cup\mybeta) - H(\myalpha\cap\mybeta)$. Then $\groundset$
  and $\groundset^*$ can adhere to each other.
\end{proposition}
\begin{proof}
See Theorem 1 in \cite{Matus07adhesivity}.
\end{proof}

\begin{corollary}\label{cor:thmadhere}
  Let $\groundset = \{X,Y,Z\}$ be a set of pseudo-variables, such that
  $Z$ is a function of $X,Y$ and $X$ is a function of $Y,Z$. Let
  $\groundset^*$ be another set of pseudo-variables such that $\groundset$ and $\groundset^*$ coincide over
  $\groundset \bigcap  \groundset^* = \{X,Y\}$. Then $\groundset^*$ and $\groundset$
  can adhere to each other.
\end{corollary}
\begin{proof}
  It is easy to verify that $\{X,Y\}$ and $\{Y,Z\}$ cannot be flats of
  $\groundset$. To prove the corollary, it suffices to prove that
  (\ref{eqn:flat}) is satisfied for all flats of $\groundset$.

  Suppose that $\myalpha$ and $\mybeta$ are flats of $\groundset$. If
  either $\myalpha$ or $\mybeta$ is the empty set, $\{Z\}$ or
  $\{X,Y,Z\}$, then either $\groundset^\prime \cap\myalpha
  \subseteq\groundset^\prime \cap\mybeta$ or $\groundset^\prime
  \cap\mybeta\subseteq\groundset^\prime \cap\myalpha $. As a result,
  $\Delta(\groundset^\prime \cap \myalpha, \groundset^\prime \cap
  \mybeta)=0$ and (\ref{eqn:flat}) holds. On the other hand, if both
  $\myalpha$ and $\mybeta$ are subsets of $\{X, Y\}$, then it is
  obvious that (\ref{eqn:flat}) remains true.  Now, suppose $\myalpha
  = \{X,Z\}$. Then (\ref{eqn:flat}) holds for $\mybeta = \{X\}$ or
  $\{X,Z\}$. Finally, when $\myalpha = \{X,Z\}$ and $\mybeta = \{Y\}$,
  by direct verification, (\ref{eqn:flat}) still holds.  Combining all
  the cases, we see that (\ref{eqn:flat}) indeed holds for all flats of
  $\groundset$.
\end{proof}
Corollary \ref{cor:thmadhere} directly leads to the following result.
\begin{theorem}\label{thm:adhere}
  Let $\groundset^*\supseteq \{X,Y\}$.  Then one can adhere the
  pseudo-variable $Z=J_{X|Y}$ to $\groundset^*$.

  If in addition $H(X)=H(Y)$, it is possible adhere a pseudo-variable
  $Z=X\oplus Y$ to $\groundset^*$.
\end{theorem}

\subsection{Proof for direct part of Theorem \ref{thm:LPbd}}
\label{sec:proofLPbd}
\begin{proof}
  To prove the direct part, we must exhibit a set of
  pseudo-variables satisfying the set of (in)equalities
  (\ref{eqn:outerbd}). Our construction works as follows:

  \begin{itemize}
  \item
  Let $ V_1, \dots, V_N$ be pseudo-variables whose pseudo-entropy function is $h$.

  \item
  By Lemma \ref{lemm:functionadhesion}, we can adhere $\sRV_{[\N]}\defined J_\groundset$ to $\groundset=\{
  V_1, \dots, V_n\}$.

  \item
  For any non-empty subset $\alpha$ of $\N$, let  $\sRV_{[\N]}$ be a pseudo-variable whose pseudo-entropy is $H(V_\alpha)$.

  \item
  By Lemma \ref{lemm:indpendentadhesion}, we adhere independent pseudo-variables $\sRV_{[\alpha]}$ to the current set of pseudo-variables $\{V_1, \dots, V_N, \sRV_{[\N]} \}$.

  \item
  By Theorem \ref{thm:adhere}, we can further adhere
  auxiliary pseudo-variables such as $J_{V_\alpha}$,
  $J_{\sRV_{[\N]}|J_{V_\alpha}}$, $J_{V_\alpha} \oplus \sRV_{[\alpha]}$
  etc.
  \end{itemize}

  Now, we will show how to associate pseudo-variables to edges.  If
  the edge is uncapacitated, then the associated pseudo-variable is
  the join of the set of pseudo-variables incident to that edge. It
  remains to show that for the three subnetworks, we can adhere
  pseudo-variables meeting all the constraints of the LP bound.

  Consider type 0 subnetworks. Let $W = \sRV_{[\alpha]}$. Then, (\ref{eqn:outerbd}) clearly
  holds. In type 1 subnetworks let $W=J_{\sRV_{[\N ]} | J_{V_\alpha}}$ and $W^\prime = J_{V_\alpha}$. Again,
  (\ref{eqn:outerbd}) holds.  Finally, for type 2 subnetworks, let $W
  = \sRV_{[\alpha]} \oplus J_{V_\alpha}$,
    $W^\prime = J_{\sRV_{[\N ]} | J_{V_\alpha}}$,
  $W^{\prime\prime} =   J_{J_{V_\alpha}|V_i}$, and
  $W^* = W^{**} = J_{V_\alpha}$. By direct
  verification, the set of (in)equalities (\ref{eqn:outerbd}) holds.
\end{proof}

\newpage

\end{document}